\documentclass[a4paper,
USenglish,
cleveref,
autoref,
thm-restate
]{lipics-v2021} 

\usepackage{booktabs}
\usepackage[textsize=tiny,disable]{todonotes}
\usepackage[normalem]{ulem}
\DeclareMathOperator{\dist}{d}
\DeclareMathOperator{\ecc}{ecc}
\DeclareMathOperator{\diam}{diam}

\DeclareMathOperator{\maxdist}{maxdist}
\DeclareMathOperator{\parent}{parent}
\newcommand{\upper}{u}

\usepackage{mathtools}
\DeclarePairedDelimiter{\parens}{\lparen}{\rparen} 
\DeclarePairedDelimiter{\braces}{\lbrace}{\rbrace} 
\DeclarePairedDelimiter{\brackets}{\lbrack}{\rbrack} 

\DeclarePairedDelimiter{\ceil}{\lceil}{\rceil}

\newcommand{\Ex}[1]{\mathrm{E}\brackets*{#1}}
\renewcommand{\Pr}[1]{\mathrm{Pr}\brackets*{#1}}

\usepackage[vlined]{algorithm2e}

\usepackage{cleveref}  

\newtheorem{property}{Property}
\Crefname{property}{Property}{Properties}

\Crefname{observation}{Observation}{Observations}
\Crefname{point}{Point}{Points}

\newcommand{\removed}[1]{}

\definecolor{mygray}{RGB}{153,153,153}
\definecolor{mydarkorange}{RGB}{213,98,0}
\definecolor{myskyblue}{RGB}{86,180,233}
\definecolor{myblue}{RGB}{0,114,178}
\definecolor{myyellow}{RGB}{240,228,66}
\definecolor{myorange}{RGB}{230,159,0}
\definecolor{mybluegreen}{RGB}{0,158,115}
\definecolor{mypink}{RGB}{204,121,167}
\presetkeys{todonotes}{color=mygray!30}{}
\usepackage{tcolorbox}

\usepackage{uniinput}

\usepackage{pgfplots}
\pgfplotsset{compat=1.18}
\definecolor{oiBlue}{HTML}{0072B2}
\definecolor{oiOrange}{HTML}{E69F00}
\definecolor{oiGreen}{HTML}{009E73}

\usepackage{fnpct}

\nolinenumbers 

\pdfoutput=1 
\hideLIPIcs  

\title{Diameter Computation on (Random) Geometric Graphs}
\author{Thomas Bläsius}{Karlsruhe Institute of Technology, Germany \and \url{http://scale.iti.kit.edu} }{thomas.blaesius@kit.edu}{https://orcid.org/0000-0003-2450-744X}{}

\author{Annemarie Schaub}{Karlsruhe Institute of Technology, Germany}{annemarie.schaub@outlook.de}{}{}

\author{Marcus Wilhelm}{Karlsruhe Institute of Technology, Germany \and \url{http://scale.iti.kit.edu} }{marcus.wilhelm@kit.edu}{https://orcid.org/0000-0002-4507-0622}{funded by the Deutsche Forschungsgemeinschaft (DFG, German Research Foundation) – 524989715}

\authorrunning{J. Open Access and J.\,R. Public} 

\Copyright{Jane Open Access and Joan R. Public} 



\category{} 

\relatedversion{} 

\acknowledgements{The authors thank Tillmann Bühler for helpful
  discussions.}

\EventEditors{}
\EventNoEds{0}
\EventLongTitle{}
\EventShortTitle{}
\EventAcronym{}
\EventYear{2025}
\EventDate{}
\EventLocation{}
\EventLogo{}
\SeriesVolume{}
\ArticleNo{}

\begin{document}

\maketitle


\begin{abstract}
  We present an algorithm that computes the diameter of random
  geometric graphs (RGGs) with expected average degree
  $\Theta(n^{\delta})$ for constant $\delta\in(0,1)$ in
  $\tilde{O}(n^{\frac{3}{2}(1+{\delta})} +n^{2 - \frac{5}{3}{\delta}})$
  time, asymptotically almost surely. This brings the running time
  down to $\tilde{O}(n^{\frac{33}{19}})\approx \tilde{O}(n^{1.737})$
  for average degree $\Theta(n^{\frac{3}{19}})$. To the best of our
  knowledge, this constitutes the first such bound for RGGs and for a
  substantial range of average degrees, it is notably smaller than the
  recent bound of $O^*(n^{2-\frac{1}{18}}) \approx O^*(n^{1.944})$ by
  Chan, Chang, Gao, Kisfaludi-Bak, Le, and Zheng (FOCS 2025) for the
  more general class of all unit disk graphs. Our algorithm also works
  on RGGs with the flat torus as ground space, with a running time in
  $\tilde{O}(n^{\frac{3}{2}(1+{\delta})} + n^{2 - \frac{1}{3}{\delta}})$.

  While our bounds on random geometric graphs are interesting in their
  own right, they are only an application of our main contribution: A
  general framework of deterministic graph properties that enable
  efficient diameter computation. Our properties are based on the
  existence of balanced separators that are in a certain sense
  well-behaved regarding the metric space defined by the graph.  These
  properties can be seen as a distillation of the combinatorial
  features a graph gets from having an underlying geometry.

  As a by-product of verifying that RGGs fit into our framework, we
  also derive running time bounds for iFUB, a diameter algorithm by
  Crescenzi, Grossi, Habib, Lanzi, and Marino (TCS 2013) that is
  highly efficient on real-world graphs. We show that a.a.s.\ iFUB
  achieves a speedup in $\tilde{\Omega}(n^{{\delta}/3})$ over the
  naive $O(nm)$ algorithm, but runs in $\Omega(nm)$ time on torus
  RGGs. This constitutes the first theoretical analysis in a geometric
  setting and confirms prior empirical evidence, thus suggesting
  geometry as a reasonable model for certain real-world inputs.
\end{abstract}

\newpage

\section{Introduction}%
\label{sec:intro}

The \emph{diameter}, i.e., the maximum distance between any pair of
vertices, is one of the most fundamental graph parameters.  It is
relevant for numerous applications for example in network
design~\cite{MR846852,parhami2000network,DBLP:journals/jsac/XuKY04},
distributed
systems~\cite{1599738,DBLP:conf/conext/ChaintreauMMD07,10.1145/863955.863999}
and graph clustering~\cite{SCHAEFFER200727}.  A simple algorithm to
compute the diameter of a graph is to perform a breadth-first search
(BFS) from every vertex, taking $O(nm)$ time on a graph with $n$
vertices and $m$ edges.  The iFUB algorithm (short for iterative
fringe upper bound)~\cite{ifub} constitutes a notable improvement over
this approach in practice.  Despite a $Θ(nm)$ worst-case running time,
it is often much faster on real-world inputs, especially on complex
scale-free networks~\cite{axiomatic_borassi_2017} and graphs with
underlying geometry~\cite{external_validity}.

The core intuition behind iFUB is that on many real-world networks
there is a meaningful notion of center and periphery (or fringe).
More precisely, vertices in the center have smaller distances to most
other vertices than vertices in the periphery.  Moreover, distant
pairs of vertices always lie in the periphery and their shortest paths
are (roughly) bisected by the center.  The iFUB algorithm exploits
this structure by heuristically choosing a vertex in the center and
then restricting the search for diametric vertices to vertices that
are sufficiently far from this center.  As a result, the algorithm
only executes a constant number of breadth-first searches in order to
select a central vertex and afterwards only performs a BFS for
vertices that have distance at least half the diameter from this
central vertex.  On graphs with a strong center--periphery structure,
where diametric paths are indeed approximately halved, this results in
a low number of BFS runs.

An extensive empirical study confirms that many real-world networks
exhibit a sufficiently pronounced center--periphery structure for iFUB
to achieve sublinear running times in
practice~\cite{external_validity}.  In particular, the study
identifies two regimes of such networks.  The first consists of graphs
with a strongly heterogeneous degree distribution, i.e., scale-free
networks.  For this setting Borassi, Crescenzi, and
Trevisan~\cite{axiomatic_borassi_2017} prove running time bounds for
iFUB under the assumption of a power-law degree distribution together
with independently sampled edges.  The second regime are graphs with a
homogeneous degree distribution and high locality, i.e., with an
underlying geometric structure.  While it might seem intuitive that
such graphs are benign for iFUB, the authors
of~\cite{external_validity} also identify some notable exceptions.
These are graphs with a clearly apparent geometric structure, but a
periodic geometric ground space, where distances ``wrap around'' like
on a flat torus or a spherical surface.  Intuitively, such graphs have
neither center nor periphery, thus any chosen central vertex splits
some diametric paths very unevenly, and iFUB needs to explore a large
portion of the graph in order to find the diameter.  To the best of
our knowledge, the performance of iFUB on geometric graphs has not yet
been formally analyzed.  In particular, the conjecture that it benefits
from aperiodic geometry but deteriorates on periodic geometry has not
been studied from a rigorous theoretical perspective.

In this work, we provide the first theoretical explanation of this
behavior by analyzing iFUB's performance on random geometric graphs
(RGGs).  We show that on RGGs with a square ground space, iFUB
achieves an asymptotic speed-up when choosing a central vertex using
the so-called 2-sweep heuristic and that this is not the case on RGGs
with a flat torus as ground space.
We say an event occurs \emph{asymptotically almost surely} (a.a.s.) if
its probability is at least $1 - o(1)$ and \emph{with high
  probability} if its probability is at least $1 - O(\frac{1}{n})$.

\removed{
\begin{figure*}
  \centering
  \includegraphics[width=0.3\textwidth]{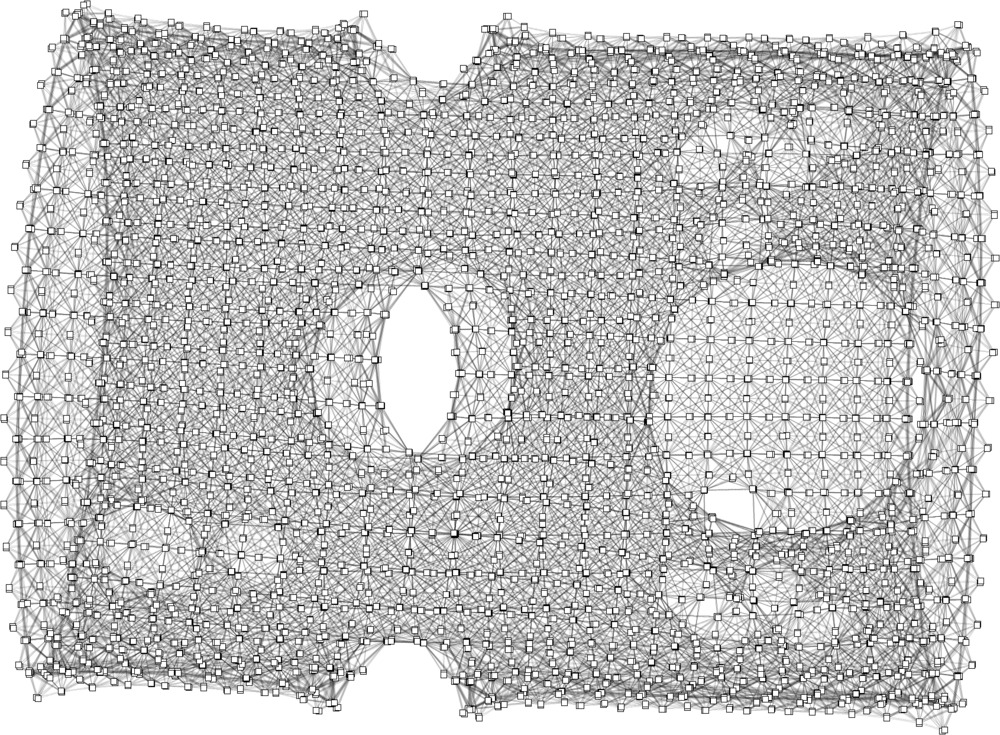}
  ~
  \includegraphics[width=0.3\textwidth]{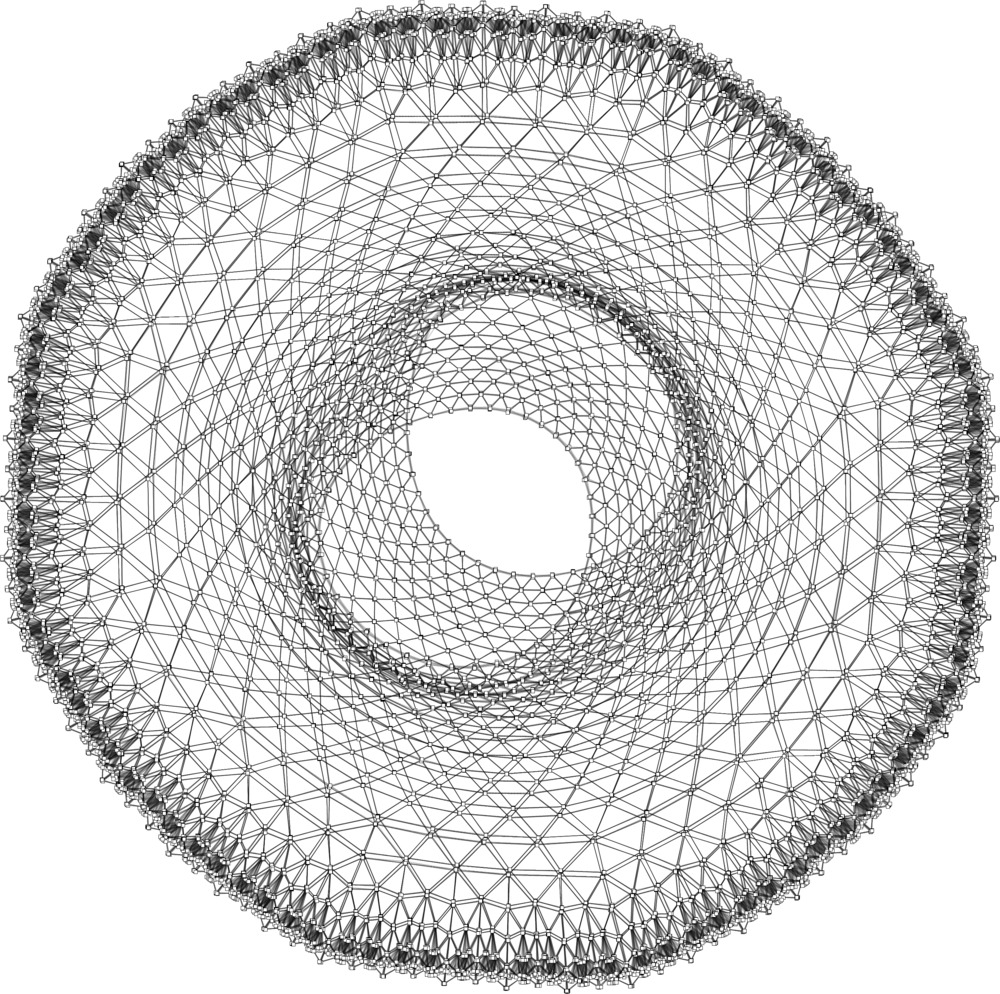}
  ~
  \includegraphics[width=0.3\textwidth]{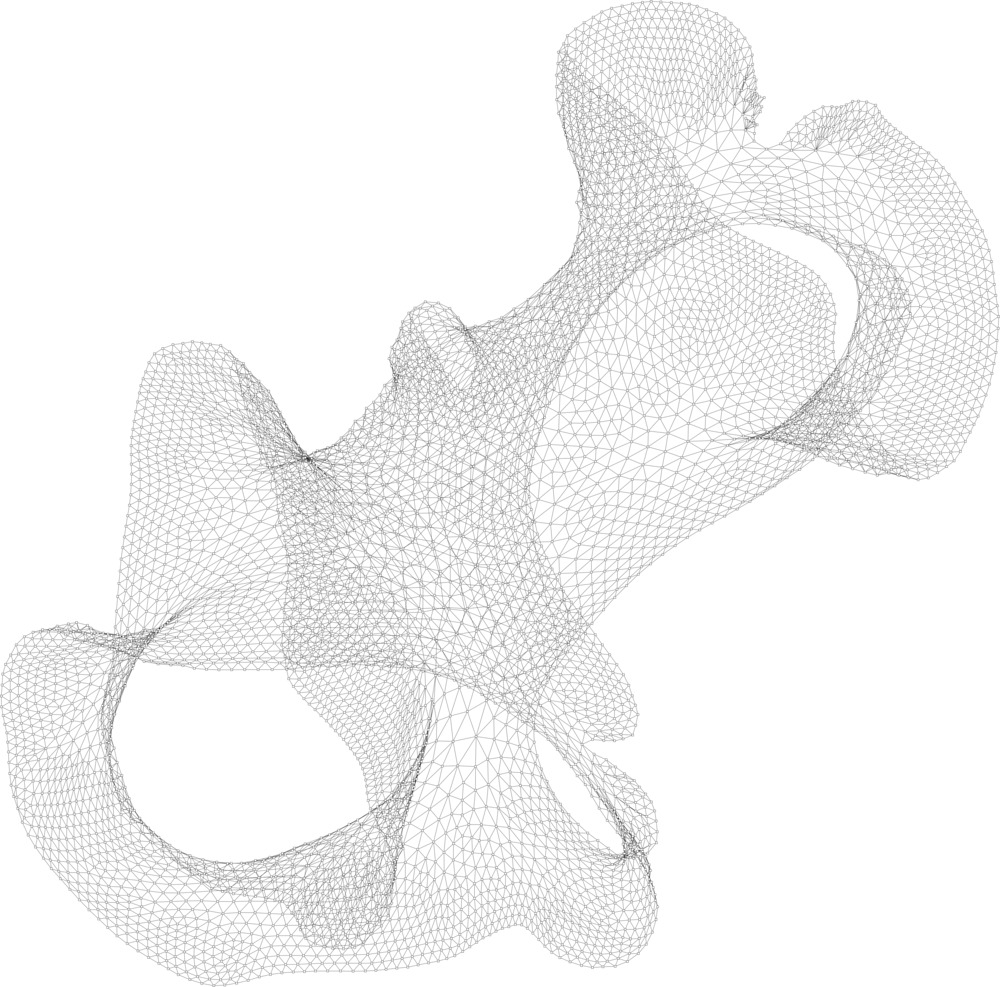}
  \caption{Drawings of three graphs identified as particularly hard
    inputs for iFUB despite their apparent geometric structure; see
    \cite[Figure 6]{external_validity}.}
  \label{fig:periodic_graphs}
\end{figure*}
}

\begin{restatable}{theorem}{ifubfast}\label{thm:ifub_fast}
  Let $G\sim\mathcal{G}(\mathcal{S}, n, r)$ be a square random
  geometric graph with expected average degree
  $d∈Ω(\log^{\frac32} n)$.  Then, a.a.s., 2-sweep iFUB has running
  time in $O((nd^{-\frac23}+ \log n)m)$.  If
  $G\sim\mathcal{G}(\mathcal{T}, n, r)$ is a torus RGG with expected
  average degree in $Ω(\log^{\frac32} n)$, then a.a.s.\ for every
  choice of the central vertex, the running time of iFUB is in
  $Ω(nm)$.
\end{restatable}

This raises the question of whether the resulting running time in
$Ω(nm)$ is inherent to graphs with a torus-like geometry or whether
better algorithmic approaches are possible.  In the following theorem,
we give a positive answer by showing that a polynomial improvement to
this is possible.  Still, for square RGGs we achieve an even better
running time exponent.\footnote{We use $\tilde{O}$-notation to hide
  poly-logarithmic factors in the running time and $O^*$ for
  $n^{o(1)}$ factors.}

\begin{figure}
  \centering
  \begin{tikzpicture}
    \begin{axis}[
      axis lines = left,
      xlabel = \(δ\),
ylabel = {exponent $x$ of running time \(\tilde{O}(n^x)\)},
      ymin=1.0,
      ymax=3.05,
      legend pos=outer north east,
      legend style={font=\small, fill=none, draw=none},
      clip=false 
    ]
      \addplot [
          domain=0:1, 
          samples=3,
          color=gray,
          thick
          ]
          {2+x};
      \addlegendentry{Naive}
      \addplot [
          domain=0:1, 
          samples=3,
          color=gray,
          thick,
          dashed
          ]
          {2.373};
      \addlegendentry{Seidel's algorithm~\cite{seidel}}
      \addplot [
          domain=0:1, 
          samples=3,
          color=oiBlue,
          thick
          ]
          {2-1/8+x};
      \addlegendentry{UDG (no coords.)~\cite{chan2025trulysubquadratictimealgorithms}}
      \addplot [
          domain=0:1, 
          samples=3,
          color=oiBlue,
          dashed,thick
          ]
          {2-1/18};
      \addlegendentry{UDG (coordinates)~\cite{chan2025trulysubquadratictimealgorithms}}
      \addplot [
          domain=0:1,
          samples=23,
          color=oiOrange,
          very thick
      ]
      {max(3/2+1.5*x,2-1/3*x)};
      \addlegendentry{\cref{thm:algo_rgg} (torus)}
      \addplot [
          domain=0:1,
          samples=39,%
          color=oiOrange,
          dotted,
          very thick
      ]
      {max(2-5/3*x,1.5+1.5*x)};
      \addlegendentry{\cref{thm:algo_rgg} (square)}
    \end{axis}
  \end{tikzpicture}
  \caption{Running times on torus/square RGGs with expected
    average degree in $Θ(n^δ)$ for constant $δ∈(0,1)$.  Our algorithm
    is compared with the naive $O(nm)$ running time, Seidel's matrix-multiplication
    based $\tilde{O}(n^ω)$ approach, and the unit-disk graph
    algorithms from \cite{chan2025trulysubquadratictimealgorithms}
    in both variants with a geometric representation and without.}
  \label{fig:running_times}
\end{figure}
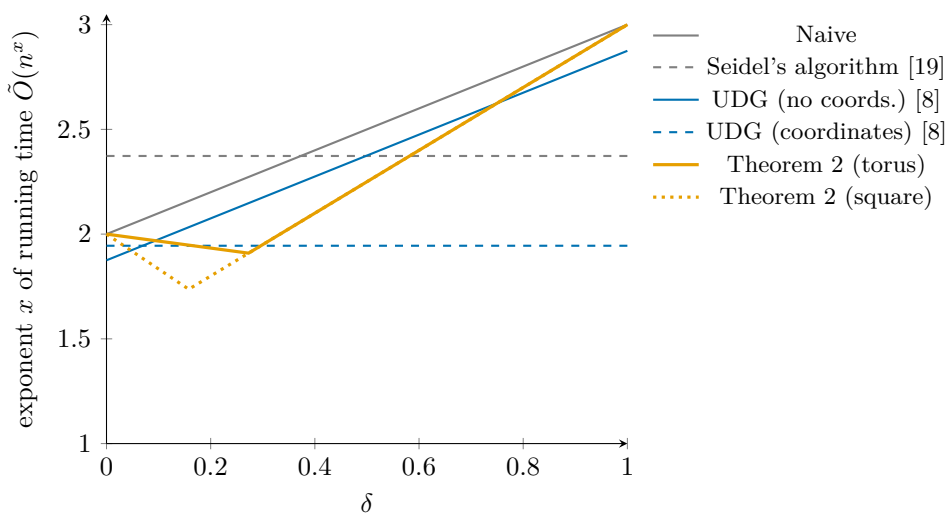

\begin{restatable}{theorem}{runningTimeRGG}\label{thm:algo_rgg}
  On RGGs with expected average degree $Θ(n^δ)$ for constant
  $δ∈(0,1)$, asymptotically almost surely the diameter can be computed
  in $\tilde{O}(n^{\frac32(1+δ)} + n^{2-\frac{1}{3}δ})$ time for torus
  RGGs, respectively
  $\tilde{O}(n^{\frac32(1+δ)} + n^{2-\frac{5}{3}δ})$ time for square
  RGGs.
\end{restatable}

We note that in the above theorem the probabilistic statement only
concerns drawing the random geometric graph; the algorithm itself is
fully deterministic and always computes the diameter correctly.  To
the best of our knowledge these are the first running time bounds for
the diameter problem specifically on random geometric graphs.  For an
overview of how they compare to known running times for diameter
computation on related graph classes, see
\cref{fig:running_times}.  Note that specifically the class of unit
disk graphs makes for a suitable comparison as it can be seen as a
deterministic worst-case variant of (square) random geometric graphs.
Here, the fastest known running time is in $O^*(n^{2-1/18})$ as shown
recently by Chan, Chang, Gao, Kisfaludi-Bak, Le, and
Zheng~\cite{chan2025trulysubquadratictimealgorithms}.  Our running
time on square RGGs gives a polynomial improvement upon this for
average degrees $Θ(n^δ)$ with constant $δ$ strictly between
$1/30 \approx 0.033$ and $8/27 \approx 0.296$.  Even our running time
on torus RGGs gives a polynomial improvement for $d$ strictly between
$1/6\approx 0.167$ and $8/27\approx 0.296$.  Additionally, we note
that the $O^*(n^{2-1/18})$ algorithm depends on a coordinate
representation of the graphs, which is $∃ℝ$-hard to
obtain~\cite{DBLP:journals/dcg/KangM12}.  A variant of their algorithm
that only requires the graph as input runs in $\tilde{O}(mn^{1-1/8})$
time~\cite{chan2025trulysubquadratictimealgorithms}.  Compared to
this, our algorithm on degree $Θ(n^δ)$ random geometric graphs is
faster for $δ ≥ \frac{3}{64} \approx 0.047$ (square RGGs),
respectively $δ≥\frac{3}{32} \approx 0.094$ (torus RGGs).  Concerning
the requirements on the input, we note that our algorithm lies between
these two variants.  We only use the coordinates implicitly in the
sense that we require a hierarchy of separators that is
straightforward to compute given coordinates.  But in principle, the
separator hierarchy could be computed in a different way without
having the coordinates as an intermediate step.


While the algorithms stated in \cref{thm:algo_rgg} are interesting
contributions on their own, they are only an application of a more
general framework.  Our main contribution is to distill a set of
deterministic graph properties and to prove that they enable efficient
diameter computation.  The running times on random geometric graphs
then follow, by proving that these graphs a.a.s.\ fit into our
framework.

\begin{figure}
  \centering
    \begin{subfigure}[b]{0.45\textwidth}
      \centering
        \includegraphics[page=9]{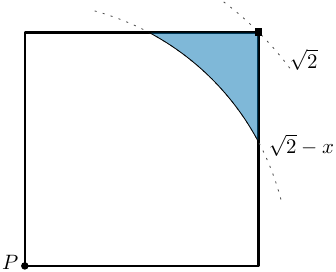}
        \subcaption{}
        \label{fig:new_properties:torus_antipodal}
    \end{subfigure}
    ~ 
    \begin{subfigure}[b]{0.45\textwidth}
      \centering
        \includegraphics[page=10]{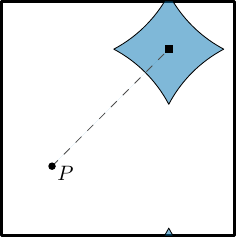}
        \subcaption{}
        \label{fig:new_properties:torus}
    \end{subfigure}
    ~ 
    \begin{subfigure}[b]{0.45\textwidth}
      \centering
        \includegraphics[page=11]{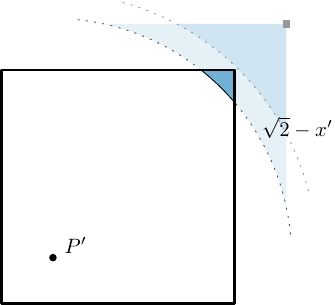}
        \subcaption{}
        \label{fig:new_properties:corners}
    \end{subfigure}
    ~ 
    \begin{subfigure}[b]{0.45\textwidth}
      \centering
        \includegraphics[page=12]{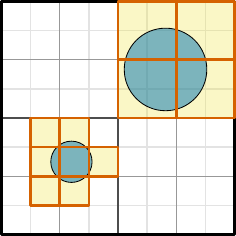}
        \subcaption{}
        \label{fig:new_properties:grid}
    \end{subfigure}
    \caption{Visualization of the observations.  Regarding
      \cref{item:geom-diam-partners}, Part (a) shows a point $P$ on
      the unit square, its unique diametric partner at distance
      $\sqrt{2}$ (black square); the blue region of all $x$-diametric
      partners of $P$ for $x = \sqrt{2} / 5$ is clearly not much
      larger than $x$.  Part (b) shows an analogous situation on the
      torus; note that here, any point has a diametric partner.  Part
      (c) visualizes the second part of \cref{item:geom-few-corners}:
      a point $P'$ sufficiently far from any corner of the unit-square
      has no $x$-diametric partners for $x=\sqrt{2}/5$; still, for a
      larger $x' > x$ the region of $x'$-diametric partners is
      non-empty.  Regarding
      \cref{item:geom-small-sep,item:geom-size-dependent-diam,item:geom-fragmentation},
      Part (d) shows a hierarchical subdivision of the square/torus
      ground space and two differently sized blue disks that intersect
      only few cells with similar diameter (marked in yellow).  }
  \label{fig:new_properties}
\end{figure}

To introduce our framework and motivate the properties it is based on,
we start by making a few obvious observations about geometric ground
spaces (specifically square and flat torus) and distances therein;
also see \cref{fig:new_properties}.
\cref{item:geom-diam-partners,item:geom-few-corners} are related to
\emph{diametric pairs}, where a pair of points is diametric if the
distance between them is the diameter of the ground space.  Moreover,
a pair is \emph{$x$-diametric} if their distance is at most $x$
smaller than the diameter.  If two points $p$ and $q$ form a
($x$)-diametric pair, we say that $p$ is a \emph{($x$)-diametric
  partner} of $q$ and vice versa.
\cref{item:geom-small-sep,item:geom-size-dependent-diam,item:geom-fragmentation}
are based on partitioning the ground space like with a quad-tree into
a hierarchy of cells.  Observe the following.

\begin{enumerate}
  \item \label[observation]{item:geom-diam-partners} For every point
        on a torus or square and $x>0$, the $x$-diametric partners lie
        within a disk of radius $O(x)$.  
        \hfill\textbf{\footnotesize(local diametric partners)}
  \item \label[observation]{item:geom-few-corners} In a square, only
        the four corners have a diametric partner.  Moreover, for
        every $x>0$ all points with $x$-diametric partners lie inside
        $4$ disks of radius $O(x)$. \hfill \textbf{\footnotesize(few
        corners)}
  \item \label[observation]{item:geom-small-sep} The boundary of each
        cell is small compared to its
        area. \hfill\textbf{\footnotesize(small separators)}
  \item \label[observation]{item:geom-size-dependent-diam} Cells with smaller
        area have smaller diameter and vice versa. \hfill
        \textbf{\footnotesize(size-dependent diameters)}
  \item \label[observation]{item:geom-fragmentation} Every disk in the
        ground space intersects only few cells that have diameter
        similar to the disk. \hfill \textbf{\footnotesize(low
        fragmentation)}
\end{enumerate}

We now translate these geometric observations into graph properties.
For this, the challenge is to strike a balance between introducing
enough flexibility in order to include a meaningful class of graphs
and being strong enough to allow algorithmic improvements.
In the following, notions related to distances refer to graph
distances, e.g., a ball of radius $r$ is the set of vertices with
graph distance at most $r$ from some central vertex.  To translate the
hierarchical partitioning of the geometric ground space, we introduce
\emph{recursive partitions} of a graph into \emph{blocks}.  We assume
that recursive partitions have constant \emph{branching factor}, i.e.,
each block has only a constant number of children.  We first state the
five properties and discuss them below.

\begin{enumerate}
  \item \label[property]{item:prop-diam-partners}
        \textbf{$\boldsymbol{d_{\mathrm{local}}}$-local diametric
        partners:}\quad For every vertex $v$ and $x>0$, the
        $x$-diametric partners of $v$ can be covered with $O(1)$ balls
        of radius $O(x + d_{\mathrm{local}})$.
  \item \label[property]{item:prop-few-corners}
        \textbf{$\boldsymbol{d_{\mathrm{corner}}}$-few corners:}\quad
        For every $x>0$, all vertices with an $x$-diametric partner
        can be covered with $O(1)$ balls of radius
        $O(x + d_{\mathrm{corner}})$.
  \item \label[property]{item:prop-small-sep}
        \textbf{$\boldsymbol{(\alpha, \beta)}$-small separators:}\quad
        The separator of each block with $k$ vertices has size
        $O(k^\alpha n^\beta)$.
  \item \label[property]{item:prop-size-dependent-diam}
        \textbf{size-dependent diameters:}\quad For all blocks $A$ and
        $B$ with diameters $D_A$ and $D_B$, $|A|∈O(|B|)$ implies
        $D_A ∈ O(D_B)$, and vice versa. \footnote{ We note that the
        intuitive interpretation of asymptotics is correct here:
        $O$-notation hides universal constants that do not depend on
        individual instances.  For
        \cref{item:prop-small-sep,item:prop-size-dependent-diam} it is
        also important that the blocks are only compared per instance
        and not across instances.  However, as this makes the formal
        definitions slightly tricky, we provide a full explanation in
        Appendix~\ref{sec:asymptotics_of_properties}. }
\item \label[property]{item:prop-fragmentation} \textbf{low
    fragmentation:}\quad Every ball of radius $r$ intersects a
  constant number of blocks of diameter $\Theta(r)$.
\end{enumerate}

Each property directly corresponds to the observation with the same
number, but differs in a few key ways.  In
\cref{item:prop-diam-partners,item:geom-diam-partners} we allow $O(1)$
balls and keep the linear dependency of their radius on $x$.
Additionally we introduce two parameters $d_{\mathrm{local}}$ and
$d_{\mathrm{corner}}$ that allow slack for small $x$.
\cref{item:prop-small-sep} uses two parameters to specify the size of
separators depending not only on the block size, but also on the size
of the whole graph.  This is important to also capture intersection
graphs of objects with size growing in $n$.  Finally,
\cref{item:prop-size-dependent-diam,item:prop-fragmentation} are the
same as \cref{item:geom-size-dependent-diam,item:geom-fragmentation},
except that we require the bounds to only hold up to constant factors.
We call a recursive partition \emph{$(α,β)$-well-spaced}, if
\cref{item:prop-small-sep,item:prop-size-dependent-diam,item:prop-fragmentation}
hold and it is \emph{balanced}, i.e., if for each block the size of
any two children differs only by a constant factor.  Intuitively, a
well-spaced recursive partition uses balanced sublinear separators
that divide the graph into roughly ball-shaped subgraphs.

We are almost ready to state our main theorem.  We say that a
recursive partition has leaf-block size at most $k_{\mathrm{leaf}}$ if
every leaf-block has at most $k_{\mathrm{leaf}}$ vertices.

\begin{restatable}{theorem}{thmPropertiesAlgo}\label{thm:properties_algo}
  Let $G$ be a $n$-vertex graph with degeneracy $d$ satisfying
  \cref{item:prop-diam-partners}.  Let $\mathcal{P}$ be a
  $(α,β)$-well-spaced recursive partition with leaf-block size
  $k_{\mathrm{leaf}}$.  For every $k \ge k_{\mathrm{leaf}}$ such that
  $k∈o(n)$ and blocks of size $\Theta(k)$ have diameter
  $\Omega(d_{\mathrm{local}})$, our algorithm computes the diameter of
  $G$ in time
  \begin{equation*}
    \tilde{O}\parens*{
    n^{1+α+β} \cdot d +
    \min\braces*{
    nkd +
    k^{2α}n^{1+2β} +
    k^{2α-1}n^{1+α+3β},\;
    kn^{1+α+β}
    }
    }.
  \end{equation*}
  If also \cref{item:prop-few-corners} holds and blocks of size
  $\Theta(k)$ have diameter $\Omega(d_{\mathrm{corner}})$, it runs in
  time
  \begin{equation*}
    \tilde{O}\parens*{
        n^{1+α+β}\cdot d +
        \min\braces*{
        k²d +
        k^{1+2α}n^{2β}
        + k^{2α}n^{α+3β},\;
        k²n^{α+β}
        }
    }.
  \end{equation*}
\end{restatable}

We note that the above bounds on the running time of our algorithm
hold for every $k$ satisfying the requirements, i.e., our algorithm
implicitly chooses $k$ such that the running time is minimized.  Note
that, unless $\alpha < 1/2$, the running time is monotone in $k$.  It
thus makes sense to think of $k$ as small as possible such that the
diameter of size $\Theta(k)$ block is still in
$Ω(d_{\mathrm{corner}})$.

\subparagraph*{Outline.}  The remainder of this paper is structured as
follows.  We introduce important definitions and notation in
\cref{sec:prelim}.  \cref{sec:algo} then presents our algorithm, while
\cref{sec:rgg-nice} contains our analysis on random geometric graphs.
We finally discuss generalizations of our parameters and directions
for future work in \cref{sec:conclusion}.

\section{Preliminaries}
\label{sec:prelim}
We use $[n] = \{1, \dots, n\}$.  Let $G = (V, E)$ be a (simple,
undirected) \emph{graph}.  We also use $V(G)$ and $E(G)$ to refer
to the vertex and edge set of $G$.  The \emph{distance} $d_G(v,w)$
between two vertices $v$ and $w$ in $G$ is the minimum length (i.e.,
number of edges) on a (simple) path from $v$ to $w$.  If there is no
path between $v$ and $w$, then $\dist_G(v,w) = \infty$.  Otherwise,
$v$ and $w$ are \emph{connected}.  A set of vertices (and by extension
a (sub-)graph) is connected if every pair of vertices is connected.
The \emph{eccentricity} of a vertex $v$, written $\ecc_G(v)$, is the
maximum distance between $v$ and any other vertex of $G$.  The
\emph{diameter} of $G$, written $\diam_G$ is the largest eccentricity
of any vertex.  We write
$\maxdist_G(A, B) = \max\{d_G(a,b) \mid a∈A, b∈B\}$ for the maximum
distance between two vertex sets $A, B \subseteq V(G)$.
We call two vertices (respectively, a path) \emph{diametric} if their
distance (respectively, its length) is $\diam_G$.  We also extend the
notions of eccentricity and diameter to \emph{metric spaces} in
general.  For a metric space $M$ consisting of a set of elements $X$
and a distance function $d_M$, we define the \emph{$M$-ball} of radius
$r$ around an element $e∈X$ as the set of elements $e'∈X$ with
$d_M(e,e') \le r$.
We say that a set of elements $Y \subseteq X$ can be \emph{covered} by
a ball of radius $r$, if there is an element $e∈X$ such that $Y$ is
contained in the ball of radius $r$ around $e$.  We omit the
subscripts $M$ and $G$ from the notation for eccentricity, diameter,
or balls/neighborhoods if the graph or metric is clear from context.
%
%
We write $N(v)$ for the (open) neighborhood of $v$ and $N_G(v,k) = \{w∈V(G) \mid d_G(v,w) = k\}$ for the $k$-neighborhood of
$v$, i.e., the set of vertices with distance exactly $k$.

Let $G$ be connected.  We define a \emph{recursive partition}
$\mathcal{P} = (T, \{B_v\}_{v∈V(T)})$ of $G$ as a rooted tree $T$
where each \emph{node} $v∈V(T)$ is associated with a connected set of
vertices $B_v \subseteq V(G)$.  Note that we call the vertices of $T$
\emph{nodes} to distinguish them from vertices of $G$.  The vertex
sets $\{B_v\}_{v∈V(T)}$ are called \emph{blocks}, i.e., $B_v$ is the
\emph{block} of $v$.  We require the blocks of the leaves of $T$ to
form a partition of $V(G)$.  Moreover, for each non-leaf node $u$ of
$T$ the block $B_u$ is equal to the union of the blocks of leaf nodes
in the subtree below $u$.  Note that this implies $B_r = V(G)$ for the
root node $r$ of $T$.  We further require each non-leaf node has at
least two children and at most a constant number of children (constant
branching factor).

Note that the tree $T$ is uniquely defined by the set of blocks and we
thus often use $\mathcal{P}$ as a set of blocks and write $B \in
\mathcal P$ for a block $B$.
We call a block $B_v \in \mathcal{P}$ a parent of a block $B_w$, if
$v$ is a parent of $w$ in $T$.  This lets us also use the relations of
parent, descendant, ancestor, and leaf for blocks of $\mathcal{P}$.
In particular, we write $\parent(B_w) = B_v$ and $\parent(w) = v$.  We
define the \emph{boundary} $S_{u}$ of a block $B_u$ as the subset of
$B_{u}$ that has neighbors in $V(G) \setminus B_{u}$ in the graph $G$.
The \emph{separator} of a block is the union of the boundaries of its
children.  Note that leaf blocks have empty separators and the root
block has an empty boundary.  We call $\mathcal{P}$
\emph{$ε$-balanced}, for $ε \ge 0$ if, for every node $u \in V(T)$ and
children $v$ and $w$ of $u$, we have $|B_v| \le (1+ε) |B_w|$.  We call
$\mathcal{P}$ \emph{balanced} if it is $ε$-balanced for $ε∈O(1)$.

\removed{
We introduce basic definitions and notational conventions used
throughout the paper.  For a set $X$ we denote the powerset by $2^X$
and the set of all $k$ element subset of $X$ by
${X \choose k} = \{x∈2^X \mid |x|=k\}$.  We use $[n] = \{1,\dots, n\}$
for the set of the first $n$ natural numbers.  For $a,b∈ℝ$, we denote
by $[a,b]$, $[a,b)$, $(a,b]$, and $(a,b)$ the usual closed, half-open,
and open intervals of real numbers.  We use $\log x$ to refer to
natural logarithm with base $e$.

\subsection{Graphs and Metrics} A \emph{graph} $G = (V, E)$ consists
of a set of \emph{vertices} $V$ and a set of \emph{edges}
$E \subseteq {V \choose 2}$.  In the context of multiple graphs, we
also use $V(G)$ and $E(G)$ to refer to the vertex and edge set of $G$.
Two vertices $v$ and $w$ are \emph{adjacent} if there is an edge
$\{v,w\}∈E$.  In this case $e$ and $v$ (respectively $w$) are
\emph{incident} to each other.  We denote the \emph{induced subgraph}
of a vertex set $X\subseteq V$ by $G[X] = (X, E \cap {X \choose 2})$.
The \emph{degeneracy} of $G$ is the smallest $d$ such that every
induced subgraph contains a vertex of degree at most $d$.

A \emph{path} from $v$ to $w$ is a non repeating sequence of vertices
$v = v_0, \dots, v_{\ell} = w$, where each pair of subsequent vertices
is adjacent.  The \emph{length} of this path is $\ell$.  The
\emph{distance} $d_G(v,w)$ between two vertices $v$ and $w$ in $G$ is
the minimum length of a path from $v$ to $w$.  If there is no path
between $v$ and $w$, then $\dist_G(v,w) = \infty$.  For two sets of
vertices $X, Y\subseteq V$ we write $d_G(X, Y)$ for the minimum length
of a path from a vertex $x∈X$ to a vertex $y∈Y$.

The \emph{eccentricity} of a vertex $v$ is the maximum distance
between $v$ and any other vertex of $G$, i.e.,
$\ecc_G(v) = \max\{d_G(v,w) \mid w∈V(G)\}$.  The \emph{diameter} of
$G$ is the largest eccentricity of any vertex of $G$, i.e.,
$\diam_G = \max\{\ecc(v) \mid v∈V(G)\}$.  Note that the diameter is
finite exactly if $G$ is connected.  We call two vertices
(respectively, a path) \emph{diametric} if their distance
(respectively, its length) is $\diam_G$.  We further write
$\maxdist_G(A, B) = \max\{d_G(a,b) \mid a∈A, b∈B\}$ for the maximum
distance between two vertex sets $A, B \subseteq V(G)$.  Note that
$\maxdist_G(V(G), V(G)) = \diam_G$.

We call two vertices \emph{connected} if they have finite distance.
Likewise, a set of vertices $X\subseteq V$ is \emph{connected} if the
vertices in $X$ are pairwise connected and a graph is connected if its
vertex set is connected.

A graph is a \emph{tree} if there is exactly one path between every
pair of vertices.  A \emph{rooted tree} is a tree $T$ together with a
designated root vertex $r∈V(T)$.  We call $v∈V(T)$ a \emph{descendant}
of $w∈V(T)$ if $w$ lies on the unique path from $v$ to the root $r$.
In this case, we call $w$ an \emph{ancestor} of $v$.  If $v$ and $w$
are adjacent, $w$ is the \emph{parent} of $v$, written
$\parent(v) = w$, and $v$ is a \emph{child} of $w$.  A \emph{leaf} is
a vertex without children.  For $v∈V(T)$ we define the \emph{subtree
  below $v$} as the subgraph induced by $v$ and the descendants of
$v$.  The \emph{branching factor} is the maximum number of children of
any vertex $v∈V(T)$.

The \emph{(open) $k$-neighborhood} of a vertex $v$ is the set of
vertices with distance exactly $k$ from $v$.  We write
$N_G(v,k) = \{w∈V(G) \mid d_G(v,w) = k\}$ for the $k$-neighborhood of
$v$.  The \emph{closed $k$-neighborhood} of a vertex $v$, $N_G[v,k]$,
is the set of vertices with distance at most $k$ from $v$.  The
$1$-neighborhood of $v$ is also simply called \emph{neighborhood} of
$v$, $N_G(v)$.

Note that the notions of eccentricity and diameter can be extended to
\emph{metric spaces} in general.  A metric space $M = (X, d_M)$
consists of a universe $X$ and a distance function
$d_M: X\times X \to ℝ$, such that for each element $e∈X$ we have
$d_M(e,e) = 0$, for any two elements $e,e'∈X$ we have
$d_M(e,e') = d_M(e',e)$ and $d_M(e,e')>0$ if $e\neq e'$, and for any
three elements $a,b,c$ we have $d_M(a,c) \le d_M(a,b) + d_M(b,c)$.
Note that $(V, d_G)$ is a metric space.  We define \emph{$M$-ball} of
radius $r$ around an element $e∈X$ as the set of elements $e'∈X$ with
$d_M(e,e') \le r$.  On a graph this is equal to the closed
neighborhood.  We say that a set of elements $Y \subseteq X$ can be
\emph{covered} by a ball of radius $r$, if there is an element $e∈X$
such that $Y$ is contained in the ball of radius $r$ around $e$.  We
omit the subscripts $M$ and $G$ from the notation for eccentricity,
diameter, or balls/neighborhoods if the graph or metric is clear from
context.

\subsection{Recursive Partitions} %
We define \emph{recursive partitions} as a generic recursive
subdivisions of graphs.  Let $G$ be connected.  We define the
\emph{recursive partition} $\mathcal{P} = (T, \{B_v\}_{v∈V(T)})$ of
$G$ as a rooted tree $T$ where each \emph{vertex} $v∈V(T)$ is
associated with a connected set of vertices $B_v \subseteq V(G)$.  To
distinguish $V(T)$ from $V(G)$, we call the vertices of $T$
\emph{nodes}.  The vertex sets $\{B_v\}_{v∈V(T)}$ are called
\emph{blocks}, i.e., $B_v$ is the \emph{block} of $v$.  We require the
blocks of the leaves of $T$ to form a partition of $V(G)$.  Moreover,
for each non-leaf node $u$ of $T$ the block $B_u$ is equal to the
union of the blocks of leaf nodes in the subtree below $u$.  Note that
this implies $B_r = V(G)$ for the root node $r$ of $T$.  We further
require each non-leaf node has at least two children.  We generally
assume the branching factor of $T$ to be in $O(1)$ (with respect to
$G$ and $\mathcal{P}$).

We call a block $B_v$ of $\mathcal{P}$ a parent of a block $B_w$, if
$v$ is a parent of $w$ in $T$.  This lets us also use the relations of
parent, descendant, ancestor, and leaf for blocks of $\mathcal{P}$.

We define the \emph{boundary}\todo{removed outer separator} $S_{u}$ of
a block $B_u$ as the subset of $B_{u}$ that has neighbors in
$V(G) \setminus B_{u}$.  The \emph{separator}\todo{removed inner
  separator} of a block is the union of the boundaries of its
children.  Node that leaf blocks have empty inner separators and the
root block has an empty boundary.

We call $\mathcal{P}$ $\varepsilon$\emph{-balanced}, for $ε \ge 0$ if,
for every node $u \in V(T)$ and children $v$ and $w$ of $u$, we have
$|B_v| \le (1+ε) |B_w|$.  We call $\mathcal{P}$ \emph{balanced} if it
is $ε$-balanced for $ε∈O(1)$ and \emph{nice} if it is balanced and has
constant branching factor.

}

\section{Diameter Algorithm}%
\label{sec:algo}

Let $G = (V, E)$ be a graph and let $\mathcal{P}$ be a recursive
partition of $G$.  Our algorithm roughly works as follows.  We choose
a \emph{flat partition}, a set of similarly sized blocks
$\mathcal{B} \subset \mathcal P$ that together form a partition of
$V$.  This can be easily achieved by choosing some parameter
$k ≥ k_{\mathrm{leaf}}$ and including a block $B \in \mathcal{P}$ in
$\mathcal{B}$ if it has size at most $k$ while its parent has size
larger than $k$, i.e.,
$\mathcal{B} = \{B \in \mathcal{P} \mid |B| \le k \text{ and }
|\mathrm{parent}(B)| > k\}$.

As $\mathcal{B}$ is a partition of $V$, computing $\maxdist_G(A,B)$
for every pair of blocks $A, B \in \mathcal{B}$ (including $A=B$)
yields the diameter of the graph.  To save time, we ignore a pair of
blocks $A, B \in \mathcal{B}$ if $\maxdist_G(A,B)$ is obviously too
small to be relevant for the diameter.  For this, we use an upper
bound $\upper(A, B)$ on $\maxdist_G(A,B)$ and a global lower bound $\ell$
on the diameter.  If $\upper(A,B) < \ell$, we can safely skip the pair
$(A, B)$.  Otherwise, we call $(A, B)$ a \emph{candidate pair} and $A$
a \emph{candidate partner} for $B$ (and vice versa).  With this, it
remains to solve the following problems.

\begin{enumerate}
  \item\label[point]{item:dist} Efficiently compute the $\maxdist$ for each
        candidate pair.
  \item\label[point]{item:pairs} Bound the number of candidate pairs and
        compute them efficiently.
  \item\label[point]{item:ell} Obtain a lower bound $\ell$.
\end{enumerate}

Regarding \cref{item:dist}, we discuss in \cref{sec:algo:preproc} how
to efficiently compute $\maxdist(A, B)$ for a candidate pair $(A, B)$.
For this, we introduce a pre-processing step in which we construct a
data structure that acts as an exact distance oracle and lets us
quickly compute the distance between arbitrary vertices.

For \cref{item:pairs}, the number of candidate pairs depends on the
upper bound $\upper(A,B)$.  We define $\upper(A, B)$ in
\cref{sec:algo:upper_bound} and show how it can be efficiently
evaluated using the distance oracle from \cref{sec:algo:preproc}.
Assuming that $G$ has $d_{\mathrm{local}}$-local diametric partners
(\cref{item:prop-diam-partners}) (and optionally also
$d_{\mathrm{corner}}$-few corners, \cref{item:prop-few-corners}) and
that $\mathcal{P}$ is $(α,β)$-well-spaced
(\cref{item:prop-small-sep,item:prop-size-dependent-diam,item:prop-fragmentation})
we then show that our definition of $\upper(A,B)$ leads to a small
number of candidate pairs.  In fact, there can be much fewer candidate
pairs than pairs of blocks.  In
\cref{sec:algo:candidate_identification} we provide a way of computing
all candidate pairs that does not need to consider all pairs of
blocks.


Finally, we address \cref{item:ell}, by essentially performing a
binary search on the solution.  For this, we treat $\ell$ not as a
lower bound on the diameter, but simply as a \emph{guess} for the
diameter.  Then, using the approach outlined above, we can decide
whether $\ell$ is smaller, equal or larger than the diameter.
Assuming $\ell ≥ \diam_G$, there are only few candidate pairs and the
algorithm terminates quickly (see also the discussion of
\cref{item:pairs} above).  Equivalently, if the algorithm does not
terminate quickly, this means that $\ell < \diam_G$.  This allows us
to determine $\diam_G$ in a binary search.  In
\cref{sec:algo:diameter} we discuss the details for this, including an
additional exponential search in order to choose $k$ optimally.


\subsection{Distance Oracle and Maxdist Computation}%
\label{sec:algo:preproc} 

In the pre-processing step, we conduct breadth-first searches from
separator vertices in the blocks of the recursive
partition.  This allows us to construct an exact distance oracle and
compute vertex eccentricities in each block.

The rough idea for the distance oracle is very simple.  If a vertex
set $S$ separates two vertices $v$ and $w$, then shortest paths
between $v$ and $w$ cross $S$.  Thus there exists a vertex $s∈S$ such
that the distance between $v$ and $w$ is $\dist_G(v,s) + \dist_G(s,w)$
and the distance between $v$ and $w$ can be found by checking the
distances to $v$ and $w$ from each $s∈S$.  We note that this is a
standard approach and has already been used for distance oracles and
(directed) reachability oracles in other
settings~\cite{DBLP:journals/algorithmica/FarzanK14,DBLP:conf/esa/ArikatiCCDSZ96,de_Berg_2023}.

\begin{lemma}\label{lem:dist-oracle}
  Let $G$ be a graph with degeneracy $d$ and a balanced recursive
  partition $\mathcal{P}$ with $(α,β)$-small separators
  (\cref{item:prop-small-sep}).  Then, in $O(n^{1+α+β} \cdot d)$ time,
  we can construct a data-structure $\mathcal{D}$ requiring
  $O(n^{1+α+β})$ space that can
  \begin{romanenumerate}
    \item for any two vertices $v,w$ compute $\dist_G(v,w)$ in
    $O(n^{α+β})$ time, unless $v$ and $w$ are both non-boundary
    vertices of the same leaf-block of $\mathcal{P}$,
    %
    %
    \item for each block $B∈\mathcal{P}$ return a vertex $b∈B$ and
    $\ecc_{G[B]}(b)$ in $O(1)$ time.
  \end{romanenumerate}
\end{lemma}
\begin{proof}
  For each block $B$ of $\mathcal{P}$ we perform a BFS restricted to
  $G[B]$ from each vertex $s$ in the separator of $B$.  We store the
  distances from $s$ to each vertex $b∈G[B]$ in an array $D_{B,s}$.
  For the running time, note that a block $B$ with $k$ vertices has
  $k \cdot d$ edges and an separator of size $s_k ∈ O(k^α\cdot n^β)$.
  Thus, the cost for the $s_k$ BFS runs on $G[B]$ is in
  $O(s_k \cdot kd)$.  We write $T(k)$ for the running time of these
  BFS runs plus the running time in all descendant blocks of $B$.
  Each block $B$ has $c ∈ O(1)$ children that form a partition of $B$,
  so this results in the recurrence
  \begin{equation*}
    T(k) = k^α \cdot n^β \cdot k \cdot d + \sum_{i=1}^{c} T(b_i k),
  \end{equation*}
  where the $b_i∈(0,1)$ are constants summing up to $1$,
  $\sum_{i=1}^{c} b_{i} = 1$.
  %
  %
  Then, $T(n) ∈ O(n^{1+α+β} \cdot d)$ follows via induction over $n$
  or by applying the theorem of Akra and Bazzi~\cite{akra_bazzi}.
  Similarly, for each size $k$ block we need to store $O(k)$
  distances, so total space needed is in $O(n^{1+α+β})$.  Note that
  these BFS runs allow us to store, for each block $B$, the
  eccentricity of a vertex of $G[B]$ without incurring any additional
  asymptotic overhead.  This implies statement (ii).

  It remains to discuss the distance queries.  Let $v$ and $w$ be two
  vertices and let $P$ be a shortest path between them in $G$.
  Consider the smallest block $B$ that contains $P$.  If $B$ is a
  leaf-block, then both $v$ and $w$ are contained in $B$.  Then, if
  without loss of generality $v$ is a boundary vertex of $B$, we have
  $\dist_{G[B]}(v,w) = \dist_G(v,w)$.  Furthermore, $v$ is a separator
  vertex of the parent $B'$ of $B$ and the distance between $v$ and
  $w$ can be looked up in the pre-computed distance array $D_{B',v}$.
  If $v$ and $w$ are both non-boundary vertices of $B$, it $P$ may not
  contain any separator vertices, so we ignore this case.

  If otherwise $B$ is not a leaf-block, we claim that $P$ contains a
  separator of $B$.  To see this consider two cases.  If $v$ and $w$
  lie in different child blocks of $B$, then a path from $v$ to $w$
  clearly needs to cross the separator of $B$.  Otherwise, if $v$ and
  $w$ lie in the same child block $B'$ of $B$, then by the choice of
  $B$ the path $P$ is not contained in $B'$.  This means that $P$
  crosses a boundary vertex $b$ of $B'$, which is a separator vertex
  of $B$.  Then, as $P$ is contained in $B$ we have
  $\dist_G(v,w) = \dist_{G[B]}(v,b) + \dist_{G[B]}(b,w)$.  Again,
  these distances can be looked up in a pre-computed distance array
  $D_{B,b}$.

  To summarize both cases, there exists a block $B$ that is a common
  ancestor of the leaf blocks containing $v$ and $w$ and a separator
  vertex $s$ of $B$ such that
  $\dist_G(v,w) = \dist_{G[B]}(v,s) + \dist_{G[B]}(s,w)$.  Thus, in
  order to answer distance queries, we can proceed as follows.  For a
  given pair of vertices $v,w∈V(G)$ we first identify the leaf blocks
  $B_v$ and $B_w$ with $v∈B_v$ and $w∈B_w$.  This can be done in
  $O(1)$ time, assuming that as an additional preprocessing step we
  iterate over all leaves of $\mathcal{P}$ in $O(n)$ time.  Next, we
  find the lowest common ancestor $B^*$ of $B_v$ and $B_w$ in
  $\mathcal{P}$ in $O(1)$ time assuming some additional $O(n)$ time
  preprocessing~\cite{lca-query}.  Let $B_1 = B^*, \dots, B_\ell$ be
  the sequence of ancestor blocks from $B^*$ to the root.  For each
  block $B_i$ with separator $S_i$ (for $i∈[\ell]$) we iterate over
  each separator vertex $s∈S_i$ and return the minimum value of
  $\dist_{G[B_i]}(s,v) + \dist_{G[B_i]}(s,w)$.  By the considerations
  above, this correctly gives the distance of $v$ and $w$ in $G$
  unless both $v$ and $w$ are non-boundary vertices of the same
  leaf-block.

  To analyze the running time recall that distances
  $\dist_{G[B_i]}(s,v)$ and $\dist_{G[B_i]}(s,v)$ can be looked up in
  $D_{B,s}$ in $O(1)$ time.  Consequently, we need to bound the number
  of the separator vertices $S_1 ∪ \dots ∪ S_\ell$ of $B^*$ and its
  ancestors.  Each block with $k$ vertices has a separator of size
  $O(k^αn^β)$ and, for a constant $c<1$, each child block has size at
  most $c$ times the size of its parent.  Using $\ell ∈ O(\log n)$, we
  have
  \[
    \sum_{i=1}^\ell |S_i| ∈ \sum_{i=0}^{\log n} O\parens*{(n \cdot c^i)^αn^β} = O(n^αn^β) \sum_{i=0}^{\log n} (c^α)^i = O(n^αn^β),
  \]
  so the oracle query can be answered in $O(n^αn^β)$ time.
\end{proof}

In addition to accelerating the upper bound evaluation (see
\cref{sec:algo:upper_bound}), this distance oracle allows us to
efficiently compute the maximum distance between any two blocks $A$
and $B$ of $\mathcal{P}$.  We present two methods for doing so.  The
first one simply computes $\maxdist(A,B)$ using $|A| \cdot |B|$ oracle
calls.  

\begin{lemma}\label{lem:maxdist_direct_oracle}
  Let $G$ be a graph with degeneracy $d$ and let $\mathcal{P}$ be a
  balanced recursive partition with $(α,β)$-small separators
  (\cref{item:prop-small-sep}).  After a pre-processing step taking
  $O(n^{1+α+β}d)$ time, for any two distinct blocks
  $A, B ∈ \mathcal{P}$ we can compute $\maxdist_G(A,B)$ in
  $O(|A| \cdot |B| \cdot n^{α+β})$ time.
\end{lemma}
\begin{proof}
  The running time follows directly from \cref{lem:dist-oracle}, by
  calling the distance oracle for each pair of vertices in
  $A \times B$.
\end{proof}

For the second method, we construct small auxiliary graphs on which
distance computations are faster than on the whole graph.  This
improves upon the simple approach in some settings, especially on
graphs with small separators.

Let $A,B\subseteq V(G)$ be two sets of vertices.  We define the
\emph{overlay graph} $H_{(A,B)}$ of $A$ and $B$ as follows.  Let $S_A$
and $S_B$ be the boundary of $A$ and $B$.  Then, we construct
$H_{(A,B)}$ by taking $G[A∪B]$ and inserting weighted edges between
all vertices $S_A ∪ S_B$.  For $s, s'∈ S_{A} ∪ S_{B}$ the edge
$\{s, s'\}∈E(H_{(A,B)})$ has weight equal to the distance of $s$ and
$s'$ in $G$, i.e., $d_{G}(s,s')$.  The following lemma shows that
distances in the overlay graph are equal to distances in $G$.

\begin{lemma}\label{lem:aux_graph_dists}
  Let $A, B \subseteq V(G)$ be vertex subsets and let $H = H_{(A,B)}$
  be the overlay graph of $A$ and $B$.  Then for any two vertices
  $u, v∈A∪B$ their distance in $G$ is equal to their distance in $H$,
  i.e., $d_{G}(u,v) = d_{H}(u,v)$.
\end{lemma}
\begin{proof}
  We first show $d_{H_{(A,B)}}(u,v) \le d{G}(u,v)$.  For this, let $P$
  be a shortest path from $u$ to $v$ in $G$.  If $P$ contains no
  vertices of $V(G) \setminus V(H_{(A,B)})$, then $P$ is also a path
  in $H_{(A,B)}$.  For the other case, we first note that only
  boundary vertices of $A$ or $B$ can have neighbors in
  $V(G) \setminus V(H_{(A,B)})$.  Now consider an inclusion-maximal
  subpath $P_{s}$ of $P$ that contains only vertices of
  $V(G) \setminus V(H_{(A,B)})$.  Then, in $H_{(A,B)}$ the two
  vertices that come before and after $P_{s}$ on $P$ are connected by
  a weighted edge of length $|P_{s}|$.  Thus, by removing all maximal
  subpaths containing only vertices of $V(G) \setminus V(H_{(A,B)})$
  from $P$ we obtain an equally long path in $H_{(A,B)}$.

  Next, we show $d_{G}(u,v) \le d_{H_{(A,B)}}(u,v)$.  Suppose
  $P$ is a shortest path between $u$ and $v$ in $H_{(A,B)}$.  Then, an
  equally long path in $G$ can be obtained by replacing any weighted
  shortcut edge of $P$ with the shortest path between the endpoints of
  that edge in $G$.
\end{proof}

To compute the maximum distance of two blocks $A$ and $B$ of
$\mathcal{P}$, we can thus compute it on $H_{(A,B)}$ instead.  The
running time then consists of the time for the construction of
$H_{(A,B)}$ plus the time for running Dijkstra's algorithm $|A|$
times.  We summarize this in the following lemma.

\begin{lemma}\label{lem:overlay_graph}
  Let $G$ be a graph with degeneracy $d$ and let $\mathcal{P}$ be a
  balanced recursive partition with $(α,β)$-small separators
  (\cref{item:prop-small-sep}).  After a pre-processing step taking
  $O(n^{1+α+β} d)$ time, for any blocks $A, B \subseteq \mathcal{P}$
  (including $A=B$) with $k = |A∪B|$ we can compute $\maxdist_G(A, B)$
  in time
  \[
    O(k² \log k + k²d + k^{1+2α}n^{2β} + k^{2α}n^{α+3β}).
  \]
\end{lemma}
\begin{proof}
  We rely on the pre-processing and distance oracle given in
  \cref{lem:dist-oracle} and construct the overlay graph $H_{(A,B)}$.
  For this, we construct the subgraph $G[A∪B]$ in $O(kd)$ time and
  then look up all distances between boundary vertices using the
  distance oracle.  Each distance lookup takes $O(n^{α+β})$ time.
  As the blocks have $O(k^αn^β)$ boundary vertices, this results in a
  total running time in
  \[
    O(kd + (k^{α}n^β)² n^{α+β}) =
    O(kd + k^{2α} n^{α+3β})
  \]
  for the construction of $H_{(A,B)}$.  By \cref{lem:aux_graph_dists},
  we can determine $\maxdist_G(A,B)$ by running Dijkstra's algorithm
  from every vertex $a$ of $A$ to find the vertex $b∈B$ most distant
  from $a$.  The overlay graph has $O(k)$ vertices and
  $O(kd + k^{2α}n^{2α})$ edges, so Dijkstra's algorithm runs in
  \[
    O(k \log k + kd + k^{2α}n^{2β})
  \]
  time and $O(k)$ such queries take
  \[
    O(k² \log k + k²d + k^{1+2α}n^{2β})
  \]
  time.  Together with the construction of $H_{(B,C)}$, this yields
  the running time claimed in the lemma statement.
\end{proof}


\subsection{Upper Bound} 
\label{sec:algo:upper_bound}%

Recall that we want to avoid computing the maxdist of pairs of blocks
that are too close to possibly contain diametric vertex pairs.  In
this section we give the upper bound used for this purpose and show
that it can be evaluated efficiently.



\begin{lemma}\label{lem:upper_bound:eval}
  Let $G$ be a graph with degeneracy $d$ and a balanced recursive
  partition $\mathcal{P}$ with $(α,β)$-small separators.  After an
  $O(n^{1+α+β}d)$ time pre-processing, we can for any given blocks
  $A,B∈P$ (including $A=B$) compute a value $\upper(A, B)$ in time
  $O(n^{α+β})$, such that
  \[
    \maxdist_G(A, B) \le \upper(A, B) \le \maxdist_G(A, B) + 2\diam_{G[A]} + 2\diam_{G[B]}.
  \]
\end{lemma}
\begin{proof}
  As a pre-processing, we rely on the distance oracle from
  \cref{lem:dist-oracle}, which is constructed in $O(n^{1+α+β}d)$
  time.  We consider blocks $A, B∈\mathcal{P}$.

  If $A = B$ we compute $\upper(A,B)$ by selecting an arbitrary
  vertex $a∈A$ and looking up the eccentricity $\ecc_{G[A]}(a)$ in
  $O(1)$ time.  Setting $\upper(A,B) = 2 \cdot \ecc_{G[A]}(a)$, we
  have
  $\maxdist_G(A,B) = \diam_{G[A]} \le \upper(A,B) \le 2 \diam_{G[A]}$.
  It remains to consider the case $A \neq B$.  We set
  \[
    \upper(A, B) = 2 \ecc_{G[A]}(a) + 2 \ecc_{G[B]}(b) + \dist_G(a,b),
  \]
  where $a∈A$ and $b∈B$ are chosen such that, using the distance
  oracle, we can look up their eccentricities in $O(1)$ time and
  compute $d_G(a,b)$ in $O(n^{α+β})$ time (see also
  \cref{lem:dist-oracle}).

  Regarding the claimed inequalities, we first show
  $\maxdist_G(A, B) \le \upper(A, B)$.  Consider maximally distant
  vertices $a^*∈A$ and $b^*∈B$, i.e.,
  $\dist_G(a^*,b^*) = \maxdist_G(A)$, and let $P$ be a shortest path
  between $a^*$ and $b^*$.  Then, any other path from $a^*$ to $b^*$
  is at least as long as $P$.  We thus consider the path $P'$ obtained
  by concatenating a shortest $a^*a$ path in $G[A]$, a shortest $ab$
  path and a shortest $bb^*$ path in $G[B]$.  Then we have
  \[
    |P| = \maxdist_G(A,B) = \dist_G(a^*,b^*) \le \dist_{G[A]}(a^*,a) + \dist_G(a,b) + \dist_{G[B]}(b, b^*) = |P'|.
  \]
  We have $\dist_G(a^*,a) \le \diam_{G[A]} \le 2 \ecc_{G[A]}(a)$.
  With an analogous estimate on $\dist_G(b^*, b)$, we thus obtain
  \[
    \maxdist_G(A,B) \le 2\ecc_{G[A]}(a) + \dist_G(a',b') + 2\ecc_{G[B]}(b) = \upper(A,B).
  \]

  For the other claimed inequality,
  $\upper(A, B) \le \maxdist_G(A, B) + 2\diam_{G[A]} + 2\diam_{G[B]}$,
  note that $\dist_G(a,b) \le \maxdist_G(A, B)$ and additionally that
  any eccentricity in a graph is at most the diameter.
\end{proof}

This means that the $\upper(A,B)$ overshoots $\maxdist_G(A,B)$ by at most a
constant multiple of the diameters of $G[A]$ and $G[B]$.  This helps
us to analyze the effectiveness of the pruning based on this upper
bound in the next section.

\subsection{Number of Candidate Pairs}%
\label{sec:algo:cand_num} 

Recall from the beginning of \cref{sec:algo}, that we consider a flat
partition $\mathcal{B} \subset \mathcal{P}$, i.e., a set of similarly
sized blocks of $\mathcal{P}$ that partitions $V$.  Moreover, we
denote $B∈\mathcal{B}$ as a candidate partner for $A$ under $\ell$, if
$\upper(A, B) \ge \ell$.  In the following, we give upper bounds for
the number of candidate pairs in $\mathcal{B}$ on graphs with local
diametric partners and, optionally, few corners.  We begin by
formalizing the intuition that candidate pairs are located far from
each other in the graph.

\begin{lemma}\label{lem:candidate_far}
  For $\ell ≥ \diam_G$, let $A$ and $B$ be two candidate pairs under
  $\ell$, i.e., $\upper(A, B) ≥ \ell$.  Then for every pair of
  vertices $a∈A$ and $b∈B$ we have
  \[\dist_G(a,b) ≥ \diam_G - 3 \diam_{G[A]} - 3 \diam_{G[B]}.\]
\end{lemma}
\begin{proof}
  Consider $A$ and $B$ with $\upper(A, B) ≥ \ell$.
  Using $\ell ≥ \diam_G$ and the maximum value of $\upper(A, B)$ from
  \cref{lem:upper_bound:eval}, this implies
  \[
    \maxdist(A, B) + 2 \diam_{G[A]} + 2 \diam_{G[B]} ≥ \upper(A, B) ≥ \ell ≥ \diam_G.
  \]
  This means there are vertices $a^*∈A$ and $b^*∈B$ with
  \[
    \dist_G(a^*, b^*) ≥ \diam_G - 2 \diam_{G[A]} - 2 \diam_{G[B]}.
  \]
  For a pair of vertices $a∈A$ and $b∈B$, we have
  \[
    \dist_G(a^*,b^*) ≤ \dist_G(a^*,a) + \dist_G(a,b) + \dist_G(b,b^*) ≤ \diam_{G[A]} + \dist_G(a,b) + \diam_{G[B]}
  \]
  Together, this implies
  $\dist_G(a,b) ≥ \diam_G - 3 \diam_{G[A]} - 3 \diam_{G[B]}$.
\end{proof}

Assume that $G$ has $d_{\mathrm{local}}$-local diametric partners
(\cref{item:prop-diam-partners}) and $\mathcal{P}$ is
$(α,β)$-well-spaced.  We show that among the similarly sized blocks
$\mathcal{B}$ each block has only few candidates.  The idea for this
is roughly as follows.  By \cref{lem:candidate_far}, vertices of a
candidate pair are almost diametrical.  However, by
\cref{item:prop-diam-partners} the almost diametrical partners of any
vertex are covered by few balls of bounded radius and by
\cref{item:prop-fragmentation} only few blocks with relevant diameters
intersect any such ball.  This gives a bound on the number of
candidates.

\begin{lemma}\label{lem:candidate_bound}
  Let $G$ be a graph with $d_{\mathrm{local}}$-local diametric
  partners and a $(α,β)$-well-spaced recursive partition
  $\mathcal{P}$.  Let further $\ell \ge \diam_G$, and let
  $\mathcal{B}⊂\mathcal{P}$ be a flat partition with similarly sized
  blocks of diameter in $Ω(d_{\mathrm{local}})$.  Then each block
  $A∈\mathcal{B}$ has only $O(1)$ candidates in $\mathcal{B}$.
\end{lemma}
\begin{proof}
  Let $A∈\mathcal{B}$ be a block that has a candidate $B∈\mathcal{B}$
  under $\ell$, i.e., $\upper(A, B) \ge \ell$.  Then, for any pair of
  vertices $a∈A$ and $b∈B$ we have
  $\dist_G(a, b) ≥ \diam_G - 3 \diam_{G[A]} - 3 \diam_{G[B]}$ by
  \cref{lem:candidate_far}.  All blocks of $\mathcal{B}$ have roughly
  the same size and by \cref{item:prop-size-dependent-diam} also
  roughly the same diameters, so this means
  $\dist_G(a,b) ≥ \diam_G - x$ for $x∈Θ(\diam_{G[A]})$.  Recall that
  we call vertices satisfying the above inequality $x$-diametric.  We
  have thus shown that the vertices of any candidate block $B$ of $A$
  are among the $x$-diametric partners of $a∈A$.

  \cref{item:prop-diam-partners} guarantees that all $x$-diametric
  partners of $a$ lie in $O(1)$ balls of radius
  $O(x + d_{\mathrm{local}}) = O(x)$.  Moreover, due to
  \cref{item:prop-fragmentation}, only a constant number of blocks
  with diameter $Θ(\diam_{G[A]}) = Θ(\diam_{G[B]})$ intersect any ball
  of radius $O(x) = O(\diam_{G[A]})$, i.e., the $x$-diametric partners
  of $a$ lie in a constant number of blocks with diameter
  $Θ(\diam_{G[A]})$.  Thus, $A$ has only a constant number of
  candidates in $\mathcal{B}$.
\end{proof}

Now assume that $G$ also has $d_{\textrm{corner}}$-few corners
(\cref{item:prop-few-corners}).  We show that in a flat partition
with blocks of sufficiently large diameter only few blocks have
candidates.

\begin{lemma}\label{lem:few_corners_few_cand}
  Let $G$ be a graph with $d_{\mathrm{corner}}$-few corners, let
  $\mathcal{P}$ be a $(α,β)$-well-spaced recursive partition, let
  $\ell \ge \diam_G$, and let $\mathcal{B}⊂\mathcal{P}$ be a flat
  partition with blocks of diameter in $Ω(d_{\mathrm{corner}})$.
  Then, only $O(1)$ blocks of $\mathcal{B}$ have candidates in
  $\mathcal{B}$.
\end{lemma}
\begin{proof}
  Let $A∈\mathcal{B}$ be a block that has a candidate $B$ in
  $\mathcal{B}$, i.e., $\upper(A,B) \ge \ell$.  By
  \cref{lem:candidate_far}, for any vertices $a∈A$ and $b∈B$ we have
  $\dist_G(a,b) \ge \diam_G - 3\diam_{G[A]} - 3\diam_{G[B]}$.

  This means that $b$ is in the
  $x = (3\diam_{G[A]} + 3\diam_{G[B]})$-diametric set of $a$.  Note
  that $\diam_{G[A]}∈Ω(d_{\mathrm{corner}})$ and thus
  $x∈Ω(d_{\mathrm{corner}})$.  We have thus shown that for some
  $x∈Ω(d_{\mathrm{corner}})$ every block $A$ with a candidate in
  $\mathcal{B}$ has at least one $x$-diametric partner.  By
  \cref{item:prop-few-corners}, the set of vertices that have
  $x$-diametric partners can be covered by $O(1)$ balls of radius
  $O(x + d_{\mathrm{corner}}) = O(x)$ in $G$.  By
  \cref{item:prop-fragmentation}, any such ball is intersected by only
  $O(1)$ blocks with diameter in $Θ(x) = Θ(\diam_{G[A]})$.  This
  implies that only $O(1)$ blocks of $\mathcal{B}$ contain vertices
  with $x$-diametric partners.
\end{proof}

Assuming that the diameter of blocks in $\mathcal{B}_i$ is in
$Ω(\min\{d_{\mathrm{local}}, d_{\mathrm{corner}}\})$ this improves the
bound on the number of candidate pairs given by
\cref{lem:candidate_bound}.  We get that only $O(1)$ blocks of
$\mathcal{B}_i$ have candidates and each only has $O(1)$ candidates.
Thus, the total number of candidate pairs is also in $O(1)$.

\subsection{Efficient Candidate Enumeration}%
\label{sec:algo:candidate_identification} 

We have shown that within a flat partition $\mathcal{B}$ with blocks
of sufficiently large diameter, every block has only $O(1)$ candidates
(see \cref{sec:algo:cand_num}).  Additionally, for each pair of blocks
we can quickly test whether they form a candidate pair (see
\cref{lem:upper_bound:eval}).  However, if the blocks of $\mathcal{B}$
consist of $Θ(k)$ vertices, there are $Θ(n/k)$ many blocks and thus
testing each of the $O(n²/k²)$ pairs is pretty expensive, especially
for small $k$.

To improve upon this, the following observation is helpful.  Consider
two blocks $A$ and $B$ that do not form a candidate pair, i.e.,
$\upper(A,B) < \ell$.  Then we have $\maxdist_G(A, B) < \ell$ and we
can ignore all vertices $a∈A$ and $b∈B$ on the search for the
diameter.  More generally, let $A' \subset A$ and $B' \subset B$ be
descendants of $A$ and $B$, respectively.  If $A$ and $B$ do not form
a candidate pair, then $\maxdist_G(A, B) < \ell$ and hence also
$\maxdist_G(A', B') < \ell$.  In this case we consider $(A',B')$ to
not be a candidate pair regardless of the actual value of
$\upper(A',B')$.

We thus go through the recursive partition in a top-down fashion,
i.e., instead of directly considering pairs of blocks in
$\mathcal{B}$, we first evaluate the upper bounds for their ancestors
in $\mathcal{P}$, starting at the root.  This way we can already
exclude pairs of blocks in $\mathcal{B}$ that have ancestors that do
not form candidate pairs.

To make the approach more precise, we maintain an \emph{intermediate
  flat partition} $\mathcal{B}_i$.  Initially, $\mathcal{B}_0$
consists of the root block of $\mathcal{P}$.  Afterwards, in step
$i ≥ 1$, we obtain $\mathcal{B}_{i}$ from $\mathcal{B}_{i-1}$ by
replacing the largest block of $\mathcal{B}_{i-1}$ with its children
in $\mathcal{P}$.  At each step, we keep track of all candidate pairs
in $\mathcal{B}_i$.  This means that when replacing a block
$B∈\mathcal{B}_{i-1}$ with its children
$B'_1, \dots, B'_b ∈ \mathcal{B}_i$, we compute the upper bound
between each child and each candidate for $B$ in $\mathcal{B}_{i-1}$.
We stop this process with a \emph{final} flat partition
$\mathcal{B}_j = \mathcal{B}$.  In \cref{sec:algo:diameter} we discuss
how $j$ is chosen in order to minimize the running time.  Before that,
we summarize important properties of the intermediate flat
partitions.

Clearly, the balance and bounded branching factor of $\mathcal{P}$
implies that after each step $i$, the blocks $\mathcal{B}_i$ are
similarly sized, i.e., for any $A, B∈\mathcal{B}_i$ we have
$|A| ∈ Θ(|B|)$.  With \cref{item:prop-size-dependent-diam}
(size-dependent diameters), this implies
$\diam_{G[A]}∈Θ(\diam_{G[B]})$.  As $\mathcal{B}_i$ forms a partition
of $V$, it consists of $Θ(n / |A|)$ blocks for $A∈\mathcal{B}_i$.  By
\cref{lem:candidate_bound}, each block $A∈\mathcal{B}_i$ has only
$O(1)$ candidates in $\mathcal{B}_i$ under lower bound
$\ell \ge \diam_G$, provided that $\diam_{G[A]}∈Ω(d_{\mathrm{local}})$
(see \cref{item:prop-diam-partners}).  This allows us to bound the
running time needed to compute the candidate pairs in $B_j$.

\begin{lemma}\label{lem:final_candidate_set}
  Assume that the flat partition $\mathcal{B}_j$ has blocks with
  diameter in $Ω(d_{\mathrm{local}})$ and $\ell ≥ \diam_G$.  Then, in
  $O(n^{1+α+β})$ time, we can compute
  $\mathcal{B}_0, \dots, \mathcal{B}_j$ and enumerate all candidate
  pairs in $\mathcal{B}_i$ for all $i∈[0,j]$.
\end{lemma}
\begin{proof}
  At each step $i\le j$ a block $B$ from an intermediate flat
  partition $\mathcal{B}_{i-1}$ is replaced with its children and
  upper bounds are evaluated in order to maintain the set of candidate
  pairs.  However by \cref{item:prop-size-dependent-diam}, for $i<j$
  the diameters of blocks in $B_i$ are asymptotically at least as
  large as the diameters of blocks in $B_j$ and thus by
  \cref{lem:candidate_bound} the replaced block $B$ only has $O(1)$
  candidate blocks in $B_{i-1}$.  As each block of a well-spaced
  recursive partition only has $O(1)$ children, this means that the
  splitting step leads to $O(1)$ upper bound evaluations, which take
  $O(n^{α+β})$ time (see \cref{lem:upper_bound:eval}).

  It thus remains to bound the number of blocks across all considered
  flat partitions $\mathcal{B}_0, \dots, \mathcal{B}_j$.  Assume the
  blocks in $\mathcal{B}_j$ have size $Θ(k)$.  Then it consists of
  $O(\frac n k)$ blocks.  All other blocks that are part of an
  intermediate flat partition $\mathcal{B}_i$ with $i<j$ are
  ancestors of a block of $\mathcal{B}_j$.  Further, every block in
  $\mathcal{P}$ has at least $2$ children.  Thus, the total number of
  unique blocks in any of the considered partitions is also in
  $O(\frac n k)$.  This means that the total running time for all
  upper bound evaluations is in $O\parens*{\frac n k \cdot n^{α+β}}$.
  In particular this is dominated by the running time $O(n^{1+α+β})$
  for the construction the distance oracle (see
  \cref{lem:dist-oracle,lem:upper_bound:eval}).
\end{proof}

\subsection{Putting Everything Together}%
\label{sec:algo:diameter} 

In this section we combine the components laid out above and complete
the algorithm.  To recap the different steps, the fundamental approach
expects a graph $G$, a recursive partition $\mathcal{P}$ and a bound
$\ell$ and proceeds as follows.  First, construct the distance oracle,
then enumerate all candidate pairs in intermediate flat partitions
$\mathcal{B}_i$ and, for the final flat partition $\mathcal{B}_j$,
compute the maxdist of each candidate pair.  This then yields the
diameter of $G$ or, if $\mathcal{B}_j$ contains no candidate pair,
that $\diam_G < \ell$.

To analyze the algorithm, we assume that $G$ has $n$ vertices,
degeneracy $d$, satisfies \cref{item:prop-diam-partners}, and that
$\mathcal{P}$ is $(α,β)$-well-spaced.  We first consider the setting
where $\mathcal{B}_j$ is chosen such that the blocks have some
specified size $Θ(k)$.  By \cref{lem:dist-oracle} the time needed to
construct the distance oracle is in
\begin{equation}
  \label{eq:running_time_phase1}
  O(n^{1+α+β} \cdot d).
\end{equation}
Assuming the blocks in $\mathcal{B}_j$ have diameter in
$Ω(d_{\mathrm{local}})$ and $\ell ≥ \diam_G$, enumerating all
candidates in $\mathcal{B}_j$ causes no asymptotic overhead
(\cref{lem:final_candidate_set}).  Next, by \cref{lem:overlay_graph}
the running time for one maxdist computation is in
\begin{equation}
  \label{eq:running_time_overlay}
  \tilde{O}(k²d + k^{1+2α}n^{2β} + k^{2α}n^{α+3β}).
\end{equation}
Alternatively, by \cref{lem:maxdist_direct_oracle} the maxdist of a
pair of distinct blocks can also be computed in time
\begin{equation}
  \label{eq:running_time_direct}
    \tilde{O}(k²n^{α+β}).
\end{equation}
To guarantee that all candidate pairs consist of distinct blocks, it
suffices to assume $k∈o(n)$.  Then, the diameters of the blocks in
$\mathcal{B}_j$ are in $o(\diam_G)$ and thus the upper bound of any
block with itself is in $o(\diam_G)$.

As $\mathcal{B}_j$ forms a partition of the vertices, it consists of
$O(n/k)$ blocks.  Recall that we assume that $G$ has
$d_{\mathrm{local}}$-local diametric partners and that the blocks in
$\mathcal{B}_j$ have diameter in $Ω(d_{\mathrm{local}})$.
Consequently, by \cref{lem:candidate_bound}, each block has $O(1)$
candidate partners under $\ell ≥ \diam_G$.  Thus, there are $O(n/k)$
candidate pairs.  If $G$ has $d_{\mathrm{corner}}$-few corners and the
blocks of $\mathcal{B}_j$ have diameters in $Ω(d_{\mathrm{corner}})$,
then only $O(1)$ blocks have candidate partners, by
\cref{lem:few_corners_few_cand}.  This means that there are only
$O(1)$ candidate pairs.  The total running time is thus given by the
preprocessing (\cref{eq:running_time_phase1}) and the maxdist
computations for each candidate pair (\cref{eq:running_time_overlay},
respectively \cref{eq:running_time_direct}).  We summarize the
algorithm as follows.

\begin{lemma}[Size-based algorithm]\label{lem:algo_size_based}
  For $G$ and $\mathcal{P}$ as above, $\ell ∈ [n],$ and $k ∈ o(n)$,
  there is an algorithm $\mathcal{A}(G, \mathcal{P}, \ell, k)$ that
  decides how $\diam_G$ compares to $\ell$.  If $\ell \ge \diam_G$ and
  $\mathcal{P}$ admits a flat partition into blocks of size $Θ(k)$
  and diameter $Ω(d_{\mathrm{local}})$ the running time of
  $\mathcal{A}(G, \mathcal{P}, \ell, k)$ is in
  \begin{align*}
    \tilde{O}\parens*{
    n^{1+α+β} \cdot d +
    \min\braces*{
    nkd +
    k^{2α}n^{1+2β} +
    k^{2α-1}n^{1+α+3β},
    \;
    kn^{1+α+β}
    }
    }.
  \end{align*}
  If additionally $G$ has $d_{\mathrm{corner}}$-few corners and the
  diameter of the blocks is in $Ω(d_{\mathrm{corner}})$ the running
  time is in
  \begin{equation*}
    \tilde{O}\parens*{
        n^{1+α+β}\cdot d +
        \min\braces*{
        k²d +
        k^{1+2α}n^{2β}
        + k^{2α}n^{α+3β},\;
        k²n^{α+β}
        }
    }.
  \end{equation*}
\end{lemma}

Assume that for $\ell ≥ n$ the above algorithm terminates $T(n,k)$
time steps.  Then, by executing the algorithm for $T(n,k)$ steps, one
can use a binary search to determine $\diam_G$.  In fact, this
strategy can be extended to also find an optimal value for the
parameter $k$, such that the final algorithm only depends on $G$ and
$\mathcal{P}$.  With the following lemma we describe this strategy in
a generic way.

\begin{lemma}\label{lem:algo_expexpbin}
  Let $\mathcal{A}(I, \ell, k)$ be an algorithm that takes as input
  some instance $I$ 
  along with two integer parameters $\ell, k ∈ [n]$ and that decides
  how $\ell$ compares to a numerical quantity $D(I)$ with $D(I)∈[n]$.
  Assume that
  \begin{alphaenumerate}
    \item for $\ell ≥ D(I)$,
    the algorithm runs in $T(n,k)$ time, and further that
    \item there is a value $k^* ∈ [n]$ such that every $k∈Θ(k^*)$
    minimizes $T(n,k)$ for every $n$ up to constant factors.
  \end{alphaenumerate}
  Then, there is an algorithm $A'(I)$ that taking an instance $I$ that
  computes $D(I)$ in time $O(T(n,k^*) \cdot \log^2 n)$.
\end{lemma}
\begin{proof}
  The core idea is that if $k^*$ and $T(n,k^*)$ are known, then a
  binary search on $\ell$ can be used to compute $D(I)$ using
  $O(\log n)$ executions of $\mathcal{A}(I, \ell, k^*)$.  To see this,
  note that by assumption $D(s)$ takes a value between $1$ and $n$ and
  if $\mathcal{A}(I, \ell, k^*)$ terminates, it decides whether
  $\ell < D(I)$, $\ell = D(I)$, or $\ell > D(I)$.  Otherwise, if
  $\mathcal{A}(I, \ell, k^*)$ does not terminate within $T(n,k^*)$
  time, then this implies $\ell < D(I)$ and $\mathcal{A}$ can be
  halted.

  It remains to find $k^*$ and $T(n, k^*)$.  For this, we use an
  exponential search, i.e., we find (up to a constant factor) the
  smallest time limit $T$ and (up to a constant factor) the smallest
  value for $k$, such that the binary search outlined above succeeds
  in $O(T)$ time.  To be more precise, we start with a constant time
  limit $T∈O(1)$ and iteratively increase it by a constant factor
  until a subroutine succeeds.  In that subroutine, we start with
  $k = 1$ and iteratively increase $k$ by a constant factor until
  either $k>n$ or a second subroutine succeeds.  The second subroutine
  tries to determine $D(I)$ using a binary search on $\ell$, by
  executing $\mathcal{A}(I, \ell, k)$ with a time limit $T$.  If
  $k∈O(k^*)$ and $T≥T(n,k^*)$, the binary search finds $D(I)$ in
  $O(T \log n)$ steps as discussed above.  If otherwise $k$ or $T$ are
  too small, then the binary search may wrongly conclude that a probed
  value of $\ell$ is smaller than $D(I)$.  In this case the binary
  search fails and larger values for $k$ are tested until either
  $D(I)$ is found or $k>n$.  In the latter case, $T$ is increased and
  the exponential search on $k$ starts again.  This means that for
  each value of $T$, the subroutine tests $O(\log n)$ values for $k$
  in $O(T \log n)$ time each.  At some point, $T$ reaches $T(n,k^*)$.
  Then, the exponential search on $k$ succeeds and the binary search
  identifies $D(I)$.  As $T$ is increased by a constant factor each
  time, the total running time is dominated by the last round and thus
  in $O(T(n, k^*) \log² n)$.
\end{proof}

We apply this to the size-based algorithm from
\cref{lem:algo_size_based} and obtain the following.

\thmPropertiesAlgo*
%
%

We note that in our case, the second logarithmic factor of
\cref{lem:algo_expexpbin} can be avoided.  To see how, note that the
running time for the maxdist computations only depends on a few
quantities known to the algorithm, such as the number of candidate
pairs, the size of the blocks, and the size of their boundaries.  This
means that for each intermediate flat partition $\mathcal{B}_i$ the
algorithm can make an (up to constant factors) tight estimate for the
time needed to compute the maxdist of all candidate pairs in
$\mathcal{B}_i$.  Thus, the algorithm can keep track of the elapsed
time until each step $i$ and calculate the cost for computing the
maxdist over all candidate pairs of $\mathcal{B}_i$.  If this estimate
for the total running time is below the given time limit $t$ the
algorithm computes the maxdists on $\mathcal{B}_i = \mathcal{B}_j$.
Otherwise it continues by considering the next flat partition
$\mathcal{B}_{i+1}$.  This continues until either the time limit is up
or until a flat partition is found for which the time limit is
sufficient.  If there is some size $k$ and a flat partition
$\mathcal{B}_i$ with blocks of size roughly $k$ such that the
size-based algorithm $\mathcal{A}(G, \mathcal{P}, \ell, k)$ from
\cref{lem:algo_size_based} runs in time $t^*$, then for any given time
limit $t≥t^*$ the above algorithm also finds a flat partition for
which it can calculate the maxdist in time $t$.  This way only one
logarithmic factor is added by the binary search.

\section{Analysis on Random Geometric Graphs}
\label{sec:rgg-nice}

In order to apply the algorithm from \cref{thm:properties_algo} on
random geometric graphs, we show that
\cref{item:prop-diam-partners,item:prop-few-corners,item:prop-small-sep,item:prop-size-dependent-diam,item:prop-fragmentation}
hold asymptotically almost surely.  We start with an overview of the
general proof ideas and also give the main intuitions for our analysis
of the iFUB algorithm.

\subparagraph{Properties 1 and 2.}  Recall from the introduction that
\cref{item:prop-diam-partners,item:prop-few-corners} are based on
simple geometric observations
(\cref{item:geom-diam-partners,item:geom-few-corners}) that
intuitively hold for the ground spaces of torus/square RGGs on a
purely geometric level.  Consequently, the main task is to transfer
these geometric intuitions to the graph setting, i.e., it remains to
show that the derived properties hold a.a.s.\ on the graphs.  As our
main tool we extend known results on the graph--geometry stretch,
i.e., the relation between geometric distance and graph
distance~\cite{diaz_stretch}.  Roughly speaking, we use that for
vertices with known geometric distance $d$ on a RGG with connection
radius $r$, the graph distance likely lies within a narrow range
around $d / r$.  This allows us to show that the almost diametric
partners of a vertex $v$ have geometric distance close to the
geometric diameter.  They are thus contained in a small geometric ball
and therefore also in a small ball in the graph.

\subparagraph{Properties 3, 4, and 5.}  Showing that the properties
related to recursive partitions
likely hold on RGGs works similarly.  Recall that these properties
are derived from
\cref{item:geom-small-sep,item:geom-size-dependent-diam,item:geom-fragmentation}
which describe intuitive geometric properties.  We thus formally
define a recursive partition based on the quadtree-like subdivision
used for these observations, see also \cref{fig:new_properties:grid}.
Then, geometrically each cell with side length $s$ has perimeter $4s$,
area $s²$, and diameter $\sqrt{2} s$.  As the number of vertices
within any polynomially sized region is highly concentrated, this
gives us that the recursive partition of the graph is balanced and has
small separators.  With the bound on the graph--geometry stretch this
also gives the claimed size-dependent diameters.  Here, we need
slightly stronger guarantees than already shown~\cite{diaz_stretch},
because we need short paths that do not only exist in the whole graph,
but also in the subgraphs induced by the recursive partition.
Finally, it is easy to formally prove a variant of
\cref{item:geom-fragmentation}, which directly implies
\cref{item:prop-fragmentation} via the concentration of the vertices.

\subparagraph{Running Time of iFUB.}  For our analysis of the iFUB
algorithm on random geometric graphs, we consider the variant of iFUB
that chooses a central vertex $c$ using the \emph{2-sweep} heuristic
as follows.  First, the algorithm performs a BFS from an arbitrary
vertex $v$ and picks a vertex $w$ in the last layer, i.e., with
maximum distance from $v$.  Then, a second BFS is performed from $w$
and the vertex $c$ is chosen half the way on a shortest path between
$w$ and a vertex $w'$ with maximum distance from $w$.  Subsequently,
iFUB performs exactly one BFS from every vertex whose distance to $c$
is more than half the diameter of $G$.  In the settings we consider,
these vertices also account for pretty much all BFS runs and thus
directly determine the total running time.

\begin{figure}
  \centering
    \begin{subfigure}[b]{0.30\textwidth}
      \centering
        \includegraphics[page=1]{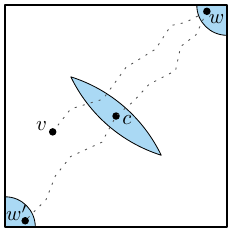}
        \subcaption{}
        \label{fig:ifub_vis:a}
    \end{subfigure}
    ~ 
    \begin{subfigure}[b]{0.30\textwidth}
      \centering
        \includegraphics[page=2]{figs/ifub_vis.pdf}
        \subcaption{}
        \label{fig:ifub_vis:b}
    \end{subfigure}
    ~ 
    \begin{subfigure}[b]{0.30\textwidth}
      \centering
        \includegraphics[page=3]{figs/ifub_vis.pdf}
        \subcaption{}
        \label{fig:ifub_vis:c}
    \end{subfigure}
    \caption{Visualization for the analysis of the iFUB algorithm.  Part (a)
      concerns the 2-sweep heuristic on square RGGs: for any chosen
      vertex $v$, the highly distant vertex $w$ lies in a small region
      (blue) around a corner of the ground space and the central
      vertex $c$ lies in a small lens (also blue, not to scale) in the
      center.  Only few vertices lie in the small regions with
      distance at least half the diameter $D$ from $c$ (Part (b),
      orange), giving an upper bound on the running time on square
      RGGs.  Part (c) shows the torus, a point $c$, and its antipodal
      partner $c'$; it can be seen that most points (orange) have
      distance at least half the diameter $D$ from $c$.}
  \label{fig:ifub_vis}
\end{figure}

For the upper bound on square RGGs, we begin by showing that the
vertex $w$ selected by the first BFS is likely located close to a
corner of the square ground space.  Afterwards, we show that $c$ is
located in a small lens close to the geometric center of the ground
space, see also \cref{fig:ifub_vis:a}.  Both of these steps use basic
geometric arguments and rely on the graph--geometry stretch, however
bounding the size of the lens in the second step is somewhat
technical.  Conditional on $c$ being located close to the geometric
center, it is then easy to show that there are not many vertices whose
distance to $c$ is at least half the diameter of $G$.  This then gives
the running time in the first part of \cref{thm:ifub_fast}.

For the torus, any chosen center $c$ results in more than half of all
points having distance more than half the diameter of $G$ from $c$,
see also \cref{fig:ifub_vis:c}.  Combined with the graph--geometry
stretch this directly shows that iFUB performs a BFS from $Ω(n)$
vertices, covering the second part of \cref{thm:ifub_fast}.

\subparagraph{Outline.}  In the remainder of this section we start
with our definition of random geometric graphs and afterwards show the
stretch bounds.  In order to apply the algorithmic framework from
\cref{sec:algo}, we then first define the recursive partition and show
that
\cref{item:prop-small-sep,item:prop-size-dependent-diam,item:prop-fragmentation}
are likely to hold, before considering
\cref{item:prop-diam-partners,item:prop-few-corners}, and finally
combining these results and apply \cref{thm:properties_algo} to get
running times for the algorithm.  Afterwards, in \cref{sec:ifub} we
analyze the running time of iFUB.

\subsection{Definitions}
\label{sec:rgg-nice:prelim}
For a side length $s∈ℝ$, we define $\mathcal{S}_s$ as the square
$[0, s)^2 \subseteq ℝ²$.  We write $d_\mathbb{E}(u,v)$ for the
Euclidean distance between two points $u,v∈ℝ²$.  Further, we define
the \emph{(flat) torus} $\mathcal{T}_s = ℝ² / (s \cdot ℤ)²$ with
side length $s$ as the equivalence classes of points in
$[0, s)²$, i.e., we write $v$ as a shorthand for
$[v] = \{v+m \mid m∈(s \cdot ℤ)²\}$ for any $v∈ℝ²$.  The
\emph{toroidal distance} between two equivalence classes is defined as
the minimum Euclidean distance between points in the equivalence
classes, i.e.,
$d_{\mathcal{T}_s}(u,v) := \min\{d_\mathbb{E}(z,w) \mid z∈[u], w∈[v]\}$.
In the context of random geometric graphs we write $\mathcal{S}$
(respectively $\mathcal{T}$) for $\mathcal{S}_{\sqrt{n}}$ (respectively
$\mathcal{T}_{ \sqrt{n} }$).

Then for a ground space $X∈\{\mathcal{T}, \mathcal{S}\}$ we define the
\emph{random geometric graph} $G∈\mathcal{G}(X, n, r)$ as follows.
Throughout this paper we assume $r < n^ρ$ for a constant
$ρ < \frac12$.  The vertex set $V(G)$ is obtained by drawing $n$
points independently and uniformly in $X$.  We identify each vertex
$v$ with its geometric position $(v_x, v_y)∈ℝ²$.  Then, two vertices
are adjacent exactly if their distance in $X$ is at most $r$, i.e.,
$E(G) = \{\{v,w\}∈{V(G) \choose 2} \mid d_X(v,w) \le r\}$.  Here, we
let $d_X$ refer to the Euclidean distance for
$\mathcal{G}(\mathcal{S}, n, r)$ and to the toroidal distance for
$\mathcal{G}(\mathcal{T}, n, r)$.  We call
$G∈\mathcal{G}(\mathcal{S},n,r)$ a \emph{square} random geometric
graph and $G∈\mathcal{G}(\mathcal{T},n,r)$ a \emph{torus} random
geometric graph.

We write $\tilde{\mathcal{G}}(X, n, r)$ for the related model of
\emph{Poisson RGGs}.  Here, we first draw a Poisson random variable
$N$ with mean $n$ and then draw $G \sim \tilde{\mathcal{G}}(X, n, r)$ as a
random geometric graph with $N$ vertices.  Note that this is
equivalent to setting $V(G)$ as the result of a Poisson point process
with intensity $1$ on $X$.  The advantage of
$\tilde{\mathcal{G}}(X, n, r)$ over $\mathcal{G}(X, n, r)$ is that the
number of vertices in any region $A\subseteq X$ with area measure $a$
follows a Poisson random variable with mean $a$ and is independent
from the number of vertices in a disjoint region
$A' \subseteq X \setminus A$.

\subsection{Graph-Geometry Stretch on RGGs}
\label{sec:rgg-nice:stretch}
There is a tight relationship between the geometric distance and the
graph distance of vertices in a random geometric graph.  For a lower
bound on the graph distance, note that in a random geometric graph
with connection radius $r$, each edge connects vertices of distance at
most $r$.  Thus, regardless of the underlying geometry
$X∈\{\mathcal{S}, \mathcal{T}\}$ for any pair of vertices $u$, $v$ we have
\begin{align}
  \label{eq:stretch-trivial}
  \dist_{G}(u,v) \ge \ceil*{\frac{d_X(u, v)}{r}}.
\end{align}

Interestingly, on random geometric graphs we also get upper bounds for
the graph distance conditional on the geometric distance of vertices.
%
%
%
Below, we slightly adapt results by Díaz, Mitsche, Perarnau, and
Pérez-Giménez~\cite{diaz_stretch} to also give upper bounds for the
graph distance within subgraphs induced by axis aligned squares.

\begin{lemma}[\cite{diaz_stretch}, Theorem 1.1]\label{lem:stretch}
  Let $G\sim\mathcal{G}(\mathcal{S}, n, r)$ be a square RGG with
  connection radius $r ∈ ω(\log^{3/4} n)$.
  Asymptotically almost surely, for every pair of vertices
  $u, v ∈V(G)$ with $d_\mathbb{E}(u,v) > r$, we have
  \begin{align*}
    \dist_{G}(u,v) \le
    \ceil*{
      \frac{d_\mathbb{E}(u,v)}{r}
      \big(1+s\big)
    }
  \end{align*}
  with
  \begin{align*}
    s = s(d_{\mathbb{E}}(u,v),r) ∈
    \begin{cases}
      O\parens*{r^{-4/3} \log n} & \text{ if $d_{\mathbb{E}}(u,v) \le r \log n$,}\\
      O\parens*{r^{-4/3}} & \text{ if $d_{\mathbb{E}}(u,v) > r \log n$}.
    \end{cases}
  \end{align*}
  Furthermore, every axis-aligned square containing $u$ and $v$
  contains a $(u,v)$-path of such length.  For
  $G∈\tilde{\mathcal{G}}(\mathcal{S}, n, r)$, the same event holds
  with probability $1-o(n^{-5/2})$.
\end{lemma}
\begin{proof}
  The upper bound on the graph distance between $u$ and $v$ is given
  in statement (ii) of Theorem~1.1 in~\cite{diaz_stretch} as
  \begin{align*}
    s = s(d_{\mathbb{E}}(u,v),r) ∈ O\parens*{
    \parens*{\frac{\log n}{r²+r\cdot d_\mathbb{E}(u,v)}}^{2/3} +
    \parens*{\frac{\sqrt{\log n}}{r}}^{4} +
    r^{-4/3}
    }.
  \end{align*}
  For $r > \log^{3/4} n$ we have
  $\parens*{\frac{\sqrt{\log n}}{r}}^{4} < r^{-4/3}$.  Thus, the
  second summand is dominated by the third one.  For the first summand
  we have
  \[
    \parens*{\frac{\log n}{r²+r\cdot d_\mathbb{E}(u,v)}}^{2/3} < \parens*{\frac{\log n}{r^2}}^{2/3} = r^{-4/3} \log^{2/3} n.
  \]
  Additionally, for $d_{\mathbb{E}}(u,v) \ge r \log n$ the first
  summand is smaller than $r^{-4/3}$.

  It remains to show that one $(u,v)$ path of such length is contained
  in a square bounding box of $u$ and $v$.  For this we need to
  consider some details made in the proof of
  Corollary~2.1~\cite{diaz_stretch}.  This corollary considers a disk
  intersection graph of a Poisson point process in the plane, where
  vertex $u$ is planted at the origin and vertex $v$ at $(t, 0)$.  The
  authors then consider the rectangle $R=[1.01α,t-1.01α]\times[0,α]$
  for $α∈Θ(r)$ with $α<0.004r$.  Relying on further lemmas that we do
  not need to discuss here, the authors show that with probability
  $1-o(n^{-5/2})$ there exists a path of the desired length from $u$
  to $v$ that only uses vertices in $R$.  The theorem then follows
  with a de-Poissonization, reducing the probability to $1-o(n^{-2})$,
  followed by a union bound over all pairs of vertices in the random
  geometric graph $G$, reducing the probability to $1-o(1)$.

  In order for the union bound to work, the authors show that the path
  within the rectangle $R$ implies a path within the
  $[0, \sqrt{n}] \times [0, \sqrt{n}]$ ground space of $G$, even if
  $u$ and $v$ lie on the boundary of the ground space.  The union
  bound afterwards does not distinguish between vertices close to the
  boundary and the many more vertices far from the boundary.  Thus,
  the proof given in \cite{diaz_stretch} also implies that for any
  pair of vertices $u$ and $v$ in $G$ there is a path of the desired
  length that does not leave any axis aligned square containing $u$
  and $v$.
%
\end{proof}

We use a coupling argument to show that \cref{lem:stretch} holds
analogously on the torus.


\begin{lemma}\label{lem:torus_stretch}
  Let $G\sim\mathcal{G}(\mathcal{T}, n, r)$ be a torus-RGG with
  connection radius $r∈ω(\log^{3/4} n)$.  Asymptotically almost
  surely, for every pair of vertices $v, w∈V(G)$ with
  $d_{\mathcal{T}}(v,w) > r$, we have
  $\dist_{G}(v,w) \le \ceil*{\frac{d_\mathcal{T}(v,w)}{r}\parens*{1+s(v,w,r)}}$
  with the error term $s(d_{\mathcal{T}}(v, w), r)$ as in
  \cref{lem:stretch}, and further, every minimal\footnote{Requiring
    \emph{minimal} squares is the main difference to the statement of
    \cref{lem:stretch}.  This difference is necessary, because on the
    torus there are squares containing $u$ and $v$ that do not contain
    the geodesic between $u$ and $v$.}  axis-aligned square containing
  $v$ and $w$ contains a $(v,w)$-path of such length.
\end{lemma}
\begin{proof}
  Consider the torus-RGG $G$.  Then, for any pair of vertices $(v,w)$
  the torus $\mathcal{T}$ can be mapped into the square $\mathcal{S}$
  such that the distance between $v$ and $w$ is preserved, i.e. such
  that the toroidal distance between $v$ and $w$ is equal to their
  Euclidean distance in the mapping.  In fact, four different mappings
  are sufficient for all pairs of vertices, see also
  \cref{fig:stretch_coupling}.
  
  More formally, we introduce a coupling between
  $\tilde{\mathcal{G}}(\mathcal{T}, n, s)$ and
  $\tilde{ \mathcal{G} }(\mathcal{S}, n, s)$ as follows.  The rough
  idea is to define four coupled RGGs by suitably re-shuffling the
  points sampled from a Poisson point process, such that the Torus
  distance between vertices is realized by the minimum Euclidean
  distance in one of the four coupled RGGs.  By \cref{lem:stretch}, in
  each of the four square RGGs graph distances are a.a.s.\ not much
  longer than implied by the geometry, so the same holds on the torus
  RGGs.

  \begin{figure}
    \centering
    \includegraphics{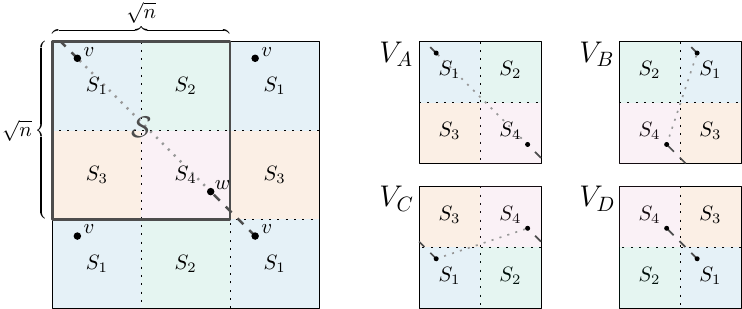}
    \caption{Sketch of the coupling argument.  The square
      $\mathcal{S}$ consists of four smaller squares and vertex sets
      $V_{A}$ to $V_{D}$ are formed by rearranging the points of a
      Poisson point process on $\mathcal{S}$ as depicted.  This way
      the torus distance between points $v$ and $w$ in $\mathcal{S}$
      is realized as the Euclidean distance in one of the
      rearrangements.\label{fig:stretch_coupling}}
  \end{figure}

  By restricting the Poisson point process $Π$ on $ℝ²$ to
  $\mathcal{S} = [0, \sqrt{n})²$ we obtain the vertex set $V_{A}$ of a
  Poisson random geometric graph $G_{A}$.  We define vertex sets
  $V_{B}$, $V_{C}$, and $V_{D}$, by translating the points sampled by
  $Π∩\mathcal{S}$ as indicated in \cref{fig:stretch_coupling}.  Let
  $S_{1}$, $S_{2}$, $S_{3}$, and $S_{4}$ be the four half-open squares
  of side length $s=\frac{\sqrt{n}}{2}$ that tile $\mathcal{S}$.
  For $i∈[1,4]$ denote $Π∩S_{i}$ as $Π_{i}$ and for a point $p∈ℝ²$
  write $τ_{p}: ℝ²\to ℝ²,\; q\mapsto q+p$ for the translation that
  maps the origin to $p$.  We define
  \begin{align*}
V_{B} &= τ_{(s,0)}\parens*{Π_{1}∪Π_{3}} \cup τ_{(-s,0)}\parens*{Π_{2}∪Π_{4}} \\
V_{C} &= τ_{(0,s)}\parens*{Π_{3}∪Π_{4}} \cup τ_{(0,-s)}\parens*{Π_{1}∪Π_{2}} \\
V_{D} &= τ_{-s,s}(Π_{4}) ∪ τ_{(s,s)}(Π_{3}) ∪ τ_{(-s,-s)}(Π_{2}) ∪ τ_{(s,-s)}(Π_{1});
  \end{align*}
  see also \cref{fig:stretch_coupling}.  Note that these vertex sets
  follow the distribution of the Poisson point process $Π$ restricted
  to $\mathcal{S}$.  Combining these vertex sets with the threshold
  radius $r$ we obtain geometric graphs $G_{B}$, $G_{C}$, $G_{D}$ that
  are each uniform Poisson random geometric graphs sampled from
  $\tilde{\mathcal{G}}(\mathcal{S}, n, s)$.

  Additionally, we define
  $G_{T} \sim \tilde{\mathcal{G}}(\mathcal{S}, n, s)$ by connecting the
  vertices of $V_{A}$ according to the torus metric.  For simplicity,
  we refer to vertices of the different graphs as $v_{i}∈V_{i}$, such
  that, for instance each $v_{A}∈V_{A}$ has a copy $v_{B}∈V_{B}$ that
  is shifted along the $x$-axis by $\frac{\sqrt{n}}{2}$ either to the
  left or to the right.  Then for any pair of vertices
  $v_{A}, w_{A}∈V_{A}$ we have
  \[
    d_\mathcal{T}(v_A,w_A) = \min_{i∈\{A,B,C,D\}}d_\mathbb{E}(v_i-w_i),
    \;\text{and}\;
    \dist_{G_{T}}(v_{A},w_{A}) = \min_{i∈\{A,B,C,D\}}\dist_{G_{i}}(v_{i}, w_{i}),
  \]
  i.e., the torus distance of $v_{A}$ and $v_{B}$ is the minimum
  Euclidean distance of any of their copies in $V_{A}$, $V_{B}$, $V_C$,
  and $V_{D}$.
  Asymptotically almost surely, the stretch event $\mathcal{E}_{\mathrm{stretch}}$ of
  \cref{lem:stretch} holds on all four graphs.
  Thus a.a.s.\ for any two vertices $v_{T},w_{T}$ of $G_{T}$ we have
  \begin{align*}
    \dist_{G_{T}}(v_{T},w_{T}) &= \min_{i∈\{A,B,C,D\}}\dist_{G_{i}}(v_{i}, w_{i}) \\
    & \le \min_{i∈\{A,B,C,D\}} \ceil*{\frac{\dist_\mathbb{E}(v_i, w_i)}{r}\parens*{1+s(v_{i},w_{i},r)}},\\
    & = \ceil*{\frac{d_\mathcal{T}(v_{T}, w_{T})}{r}\parens*{1+s(v_{T},w_{T},r)}},
  \end{align*}
  where $s(v_{i}, w_{i}, r)$ is the error term from
  \cref{lem:stretch}.  We also get that at least one such path is
  contained in the smallest axis-aligned square containing $u$ and
  $v$.
\end{proof}

We refer to events from \cref{lem:stretch}, respectively
\cref{lem:torus_stretch}, as the \emph{stretch event} and denote it
with $\mathcal{E}_{\mathrm{stretch}}$.  Note in particular that
conditioning on $\mathcal{E}_{\mathrm{stretch}}$ gives an upper bound
on the graph distance of \emph{all} pairs of vertices.

\subsection{Recursive partition}
\label{sec:rgg-nice:recursive}

We now give a formal definition for our recursive partition and show
that with high probability it is balanced and
\cref{item:prop-small-sep,item:prop-size-dependent-diam,item:prop-fragmentation}
hold.  We first define the infinite \emph{quadtree partition}
$\mathcal{P}$ as a recursive subdivision of the of the square
$\mathcal{S}=[0,\sqrt{n})\times[0,\sqrt{n})$.  We index parts using
words $\sigma\in\{1,2,3,4\}^*$.  The root of $\mathcal{P}$ is
$p_\varepsilon=\mathcal{S}$, where $ε$ stands for the empty word.  For
any $p_\sigma=[x_1,x_2)\times[y_1,y_2)$, let $x'=\frac{x_1+x_2}{2}$
and $y'=\frac{y_1+y_2}{2}$.  Then, the children of $p_σ$ (numbered
from 1 to 4)
are 
\begin{align*}
p_{\sigma 1}&=[x_1,x')\times[y',y_2),&
p_{\sigma 2}&=[x',x_2)\times[y',y_2),\\
p_{\sigma 3}&=[x_1,x')\times[y_1,y'),&
p_{\sigma 4}&=[x',x_2)\times[y_1,y').
\end{align*}
\cref{fig:quadrant_partition} shows $p_\varepsilon$ and its children.
Clearly, $\mathcal{P}$ has constant branching factor.
\begin{figure*}
  \centering
  \includegraphics{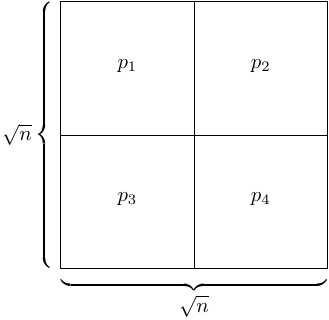}
  \caption{Four child parts $p_1$, $p_2$, $p_3$, and $p_4$ of the root
    part $p_ε$ of $\mathcal{P}$.\label{fig:quadrant_partition}}
\end{figure*}
We define the level $\ell$ of a part as the length of its index, e.g.,
$p_ε$ is on level 0 and $p_{13}$ on level 2.  Note that $\mathcal{P}$
contains $4^{\ell}$ level $\ell$ parts and each level $\ell$ part
$p∈\mathcal{P}$ has side length $\frac{\sqrt{n}}{2^{\ell}}$ and area
$\frac{n}{4^\ell}$.  For our algorithm we only need a finite subset of
$\mathcal{P}$.  For this, we define $\mathcal{P}^{\ell}$ as
$\mathcal{P}$ restricted to parts of level at most $\ell$.  We call
$\mathcal{P}^{\ell}$ the \emph{$\ell$-layered} quadtree partition.

For a geometric graph $G$ with
$V(G)\subseteq [0,\sqrt{n})\times[0,\sqrt{n})$ the partition
$\mathcal{P}^\ell$ induces a recursive partition $G[\mathcal{P}^\ell]$
of the vertices, via the intersection of $V(G)$ with the parts of
$\mathcal{P}^\ell$.  In the following, we show that for a (square or
torus) random geometric graph $G$ the recursive partition
$G[\mathcal{P}^\ell]$ a.a.s.\ is balanced, has small separators,
size-dependent diameters, and bounded fragmentation.

\subsubsection{Balance and Separator Sizes}
\label{sec:rgg-nice:rec:sep}
With respect to the area measures of its parts, the infinite quadtree
partition $\mathcal{P}$ already is balanced and has small separators,
so it remains to show the same for the induced recursive partition
$G[\mathcal{P}^\ell]$.  We use concentration bounds for the number of
vertices in regions of sufficient area to show that the induced
recursive partition of $G$ is also balanced and has small separators.

\begin{lemma}\label{lem:area_vertex_count}
  Let $G\sim\mathcal{G}(X, n, r)$ be a random geometric graph with
  ground space $X∈\{\mathcal{S}, \mathcal{T}\}$ and let
  $R \subseteq X$ be a measurable subset of $X$ with area
  $A(R) ∈ ω(\log n)$.  Then, for any constant $c$ the probability that
  the number of vertices of $G$ that lie in $R$ is between
  $A(R)\cdot \parens*{1 - \sqrt{\frac{3c\log n}{A(R)}}}$ and
  $A(R)\cdot \parens*{1 + \sqrt{\frac{3c\log n}{A(R)}}}$ is at least
  $1-O(n^{-c})$.
\end{lemma}
\begin{proof}
  Let $n_{R} = |V(G)∩R|$ be the number of vertices in $R$.  Using
  $V(G) = v_{1}, \dots, v_{n}$, we define $X_{i}$ as a Bernoulli
  random variable that indicates whether the $i$th vertex of $G$ lies
  in $R$. Then we have $n_{R} = \sum_{i=0}^{n}X_{i}$ and
  $\Ex{n_{R}} = A(R)$.  Applying Chernoff bounds (e.g., see Theorem
  4.4 and 4.5 in \cite{prob_and_comp}), for $0<δ<1$ we have
  \begin{align*}
    \Pr{n_{R} > \Ex{n_{R}} (1+δ)} &\le e^{-\frac{δ²}{3}\Ex{n_{R}}}
                                    \intertext{and}
                                    \Pr{n_{R} > \Ex{n_{R}} (1-δ)} &\le e^{-\frac{δ²}{2}\Ex{n_{R}}}.
  \end{align*}
  We set $δ = \sqrt{\frac{ 3c\log n }{\Ex{n_{R}}}}$.  Then, we have $δ<1$ as
  by assumption $\Ex{n_{p}} = A(r)∈ ω(\log n)$ and thus $δ∈o(1)$.  The
  above probabilities simplify to
  \begin{align*}
    \Pr{n_{R} > \Ex{n_{R}} (1+δ)} &\le e^{-c \log n} ∈ O(n^{-c})
    \intertext{and}
    \Pr{n_{R} > \Ex{n_{R}} (1-δ)} &\le e^{-\frac32 c \log n}
      \le e^{-c \log n} \subseteq O(n^{-c}).\qedhere
  \end{align*}
\end{proof}

We apply this to derive bounds for the size of an individual block
induced by $\mathcal{P}^\ell$.

\begin{lemma}\label{lem:whp_parts_bound}
  Let $G\sim\mathcal{G}(X, n, r)$ be a random geometric graph with
  ground space $X∈\{\mathcal{S}, \mathcal{T}\}$ and connection radius
  $r$.  For constant $α∈[0,1]$ let $\ell = α \log_{4}n$.  For every
  part $P$ of the $\ell$-layered quadtree partition $\mathcal{P}^\ell$
  with side length $s$ at least $2r$ it holds w.h.p. that the subgraph
  $G[V(G)∩P]$ induced by $P$ contains $s² (1 \pm o(1))$ vertices and
  at most $4(s-r)r(1+o(1))$ separator vertices.
\end{lemma}
\begin{proof}
  Let $P$ be a part of $\mathcal{P}^{\ell}$ with side length $s≥2r$.
  Let $G' = G[V(G)∩P]$ be the subgraph of $G$ induced by $P$ and let
  $S$ be the set of its separator vertices, i.e., the subset of
  $V(G')$ with neighbors in $G\setminus G'$.  The separator vertices
  are contained in a strip of width $r$ around the boundary of the
  square defined by $P$.  With a side length $s$ the area of $P$ is
  $s²$ and the separator vertices are contained in a region of area
  $4r(s-r)$.  We have
  $s\ge \frac{\sqrt{n}}{2^\ell} = n^{\frac{1}{2}-\frac{α}{2}}$, thus
  both areas are in $ω(\log n)$.  Thus, by
  \cref{lem:area_vertex_count} for any constant $c$ we have
  $|V(G')| = s² (1 \pm o(\frac{c\log n}{s²}))$ and
  $|S| \le 4r(s-r) (1 + o(\frac{c\log n}{4r(s-r)}))$ with probability
  $1 - O(n^{-c})$.

  The recursive partition $\mathcal{P}^{\ell}$ contains
  $O(4^{\ell}) = O(n^{α})$ parts.  Thus we can apply the union bound
  for the considered event over all $P∈\mathcal{P}^{\ell}$.  We obtain
  that the probability of the respective vertex sets being within the
  desired interval is at least $1 - O(n^{α}n^{-c})$.  For $c \ge 1+α$
  the desired event occurs with high probability.
\end{proof}

To show that $\mathcal{P}^{\ell}$ induces a balanced recursive
partition with has small separators, we apply the above lemma to every
part.

\begin{lemma}\label{lem:partition_separators}
  Let $G\sim\mathcal{G}(X, n, r)$ be a random geometric graph with
  ground space $X∈\{\mathcal{S}, \mathcal{T}\}$ and connection radius
  $r$.  Further, let $α∈(0,1)$ be a constant such that
  $r∈o(n^{1/2-α/2})$, and let $\ell = α \log_{4}(n)$.  Then, with high
  probability, the partition $G[\mathcal{P}^{\ell}]$ of $G$ induced by
  $\mathcal{P}^{\ell}$ is balanced, has $(1/2,ρ)$-small separators,
  and the leaf blocks have $Θ(n^{1-α})$ vertices.
\end{lemma}
\begin{proof}
  We condition on the event of \cref{lem:whp_parts_bound} that holds
  with high probability and show the claims one by one.
  \begin{description}
    \item[Balance.] Let $G_{σ}$ be a subgraph induced by a level
          $j=|σ|$ part $P_{σ}∈\mathcal{P}^{\ell}$, with $j<\ell$.  Let
          $G_{σ'}$ and $G_{σ''}$ be children of $G_{σ}$ induced by
          $\mathcal{P}^{\ell}$.  By \cref{lem:whp_parts_bound} $G_σ$
          has $\frac{n}{4^{j}} (1\pm o(1))$ vertices and $G_{σ'}$ and
          $G_{σ''}$ have $\frac{n}{4^{j+1}} (1\pm o(1))$ vertices.
          This means that the relative size difference of $G_{σ'}$ and
          $G_{σ''}$ tends to $1$ and thus the entire induced recursive
          partition in $ε$-balanced for arbitrarily small constant
          $ε>0$.
    \item[Small separators.] Let $G'$ be a subgraph induced by a level
          $i$ part $P∈\mathcal{P}^{\ell}$ with side length
          $s = n^{1/2}2^{-i}$ and let $S$ be the separator of $G'$.
          Then by \cref{lem:whp_parts_bound},
          $|V(G')| = s² (1\pm o(1))$ and $|S| \le 4(s-r)r (1+o(1))$.
          We have $i \le \ell = α \log_4 n$ and thus
          $s\ge n^{1/2}2^{-α \log_4 n} = n^{1/2-α/2}$.
          By assumption $r∈o(n^{1/2-α/2})$, so
          $|S|∈O(r \sqrt{|V(G')|})$.
    \item[Leaf size.] Let $G'$ be a subgraph induced by a leaf part
          $P_{σ}∈\mathcal{P}^{\ell}$.  Then by \cref{lem:whp_parts_bound}
          $|V(G')| \le \frac{n}{4^{α \log_{4} n}} (1+o(1)) = n^{1-α} (1+o(1))$.
          \qedhere
  \end{description}
\end{proof}

\subsubsection{Diameters in the Recursive Partition}
\label{sec:rgg-nice:rec:diam}
Next, we want to show that the recursive partition a.a.s.\ has
size-dependent diameters, i.e., that similarly sized blocks have
similar diameters and these are smaller for smaller blocks.  Clearly
this holds for the geometric diameters of squares, so it remains to
apply the stretch bounds from \cref{sec:rgg-nice:stretch}.
Importantly, we that for every block bounds on the diameter hold with
sufficiently high probability, such that they a.a.s.\ hold for all
blocks of the recursive partition.

Conditional on the stretch event, for every pair of vertices the graph
distance is not much larger than necessary based on the geometric
distance.  With an upper bound for the geometric diameter of each
block, this directly translates to an upper bound for the graph
diameter.  To also get a lower bound, we need to show that each block
also contains vertices with almost diametric geometric distance, i.e.,
there are vertices close to two opposite corners of the block.  To
this end, we introduce the following lemma, which shows that any
region with sufficiently large area is likely to contain at least one
vertex.

\begin{lemma}\label{lem:rgg_region_not_empty}
  Let $G∈\mathcal{G}(X, n, r)$ be a random geometric graph with $n$
  vertices and connection radius $r$ on the ground space
  $X∈\{\mathcal{S}, \mathcal{T}\}$.  Let $R \subseteq \mathcal{S}$ be a
  region with area $a$.  We have
  $\Pr{|V(G) \cap R| > 0} \ge 1-e^{-a}$.
\end{lemma}
\begin{proof}
  Let $X$ be a random variable for the number of vertices in $R$.  We
  have
  \begin{align*}
    \Pr{X>0} &= 1-\Pr{X=0}
              = 1 - \Pr{\bigcap_{v∈V} v\notin R}.
    \intertext{These events are independent and for each $v∈V$ we have
    $\Pr{v \notin R} = 1-\frac{a}{n}$.  Thus we get}
    \Pr{X>0} &= 1 - \parens*{1 - \frac{a}{n}}^{n}.
  \end{align*}
  For any $x$ we have $1+x \le e^{x}$ and thus
  $1-\frac{a}{n} \le e^{-a/n}$ and $(1-\frac{a}{n})^{n} \le e^{-a}$,
  which concludes the proof.
\end{proof}

Note that in particular a region with area $c \cdot \log n$ is
non-empty with probability at least $1 - n^{-c}$.  We use this to show
bounds on the diameter of individual blocks, first considering square-RGGs.

\begin{lemma}\label{lem:rgg_diam}
  Let $G \sim \mathcal{G}(\mathcal{S}, n, r)$ be a random geometric graph
  with connection radius $r∈ω(\log^{3/4} n)$.  Further, let
  $S\subseteq \mathcal{S}$ be an axis-aligned square of side length
  $s > r \log n$.  Then, conditional on the stretch event, we have
  \[
    \diam_{G[V(G)∩S]} \le \ceil*{\frac{\sqrt{2}s}{r} (1+Θ(r^{-4/3}))}.
  \]
  Further, for any constant $C>0$ we have with probability at least
  $1-n^{-C}$
  \[
    \diam_{G[V(G)∩S]} \ge \frac{\sqrt{2}s}{r} - 1.
  \]
\end{lemma}
\begin{proof}
  The upper bound directly follows via \cref{lem:stretch}.  Note that
  this is one of the places where we need our slightly strengthened
  version of the lemma, as we need a path using only vertices in $S$.

  For the lower bound let $x = \sqrt{C \cdot \log n + \log 2}$.  We
  consider two squares $C_{1}$, $C_{2}$ of side length $x$ located
  at two opposite corners of $S$.  The area of these squares is
  $x² = C \cdot \log n + \log 2$.  Thus, by
  \cref{lem:rgg_region_not_empty} the probability that both $C_{1}$
  and $C_{2}$ are non-empty is at least
  $1-2e^{-C \cdot \log n - \log 2} = 1 - n^{-C}$.  In this case, let
  $v_{1}∈V(G)∩C_{1}$ and $v_{2}∈V(G)∩C_{2}$ be such vertices.  Then
  the distance between these vertices is at least
  $d_{\mathbb{E}}(v_{1}, v_{2}) \ge \sqrt2 s - 2x$.  With
  \cref{eq:stretch-trivial} this means that
  \begin{align*}
    \diam_{G[V(G)∩S]} \ge d_{G[V(G)∩S]}(v_1, v_2) \ge \frac{\sqrt2 s}{r} - \frac{2x}{r}.
  \end{align*}
  By assumption $r∈ω(\sqrt{\log n})$.  Hence, $\frac{2x}{r}∈o(1)$,
  which concludes the proof.
\end{proof}

We obtain analogous bounds on the diameter of (square regions of)
torus RGGs.

\begin{lemma}\label{lem:torus_rgg_diam}
  Let $G\sim\mathcal{G}(\mathcal{T}, n, r)$ be a torus random
  geometric graph with $r∈ω(\log^{3/4} n)$.  Further, let
  $S\subseteq \mathcal{T}$ be an axis aligned square of side length
  $s$ such that $r \log n < s < \sqrt{n}$.  Then, conditional on
  the stretch event, we have
  \begin{align*}
    \diam_G \le \ceil*{\frac{\sqrt{n}}{\sqrt{2} r} (1+Θ(r^{-4/3}))}
    \intertext{and}
    \diam_{G[V(G)∩S]} \le \ceil*{\frac{\sqrt{2}s}{r} (1+Θ(r^{-4/3}))}.
  \end{align*}
  Further, for any constant $C>0$ we have with probability at least
  $1-n^{-C}$
  \begin{align*}
    \diam_G \ge \frac{\sqrt{n}}{\sqrt2r} - 1.
    \intertext{and}
    \diam_{G[V(G)∩S]} \ge \frac{\sqrt{2}s}{r} - 1.
  \end{align*}
\end{lemma}
\begin{proof}
  The geometric diameter of $\mathcal{T}$ is
  $\frac{\sqrt{2n}}{2}=\frac{\sqrt{n}}{\sqrt2}$ and with the stretch
  bounds on torus RGGs from \cref{lem:torus_stretch}, the bounds for
  $\diam_G$ follow analogously to \cref{lem:rgg_diam}.  The subgraph
  $G' = G[V(G)∩S]$ is equal in distribution an analogous subgraph of a
  square random geometric graph, as with a side length $s < \sqrt{n}$
  we avoid paths or geodesics that wrap around the torus
  $\mathcal{T}$.  Thus, the claimed bounds follow directly from
  \cref{lem:rgg_diam}.
\end{proof}

It remains to apply a union bound to show that the events from
\cref{lem:rgg_diam}, respectively \cref{lem:torus_rgg_diam}, likely
hold for each block.  In the following
\lcnamecref{lem:partition_diam}, we summarize the result for both the
square and torus setting.  Note that the geometric diameter of a side
length $s$ part $P$ of a quadtree partition is $\sqrt{2}s$ unless
$s = \sqrt{n}$ and the setting is on the torus.  In that case the
geometric diameter is $\frac{2s}{\sqrt{2}}$.

\begin{lemma}\label{lem:partition_diam}
  Let $G\sim\mathcal{G}(X, n, r)$ be a random geometric graph with ground
  space $X∈\{\mathcal{S}, \mathcal{T}\}$ and connection radius
  $r ∈ ω(\log^{3/4} n)$.  Further, let $α∈(0,1)$ be a constant such
  that $r∈o(n^{1/2-α/2})$, and let $\ell = α \log_{4}(n)$.  Then,
  asymptotically almost surely, for each subgraph $G'$ induced by a part $P$
  of the recursive partition $G[\mathcal{P}^{\ell}]$ we have
  \[
    \frac{d}{r}-1 \le \diam_{G'} \le \frac{d}{r} \cdot (1+s),
  \]
  where $d$ is the geometric diameter of $P$ in $X$ and
  $s∈O(r^{-4/3})$.  In particular, we get that
  $\diam_{G'}∈Θ(|V(G')|^{1/2}r^{-1})$, $G[\mathcal{P}^\ell]$ has
  size-dependent diameters, and leaf-blocks of $G[\mathcal{P}^\ell]$ have
  diameter in $Θ(n^{1/2 - α/2}r^{-1})$.
\end{lemma}
\begin{proof}
  The stretch event holds asymptotically almost surely on square and
  torus random geometric graphs, by \cref{lem:stretch} and
  \cref{lem:torus_stretch}.  Thus, by \cref{lem:rgg_diam},
  respectively \cref{lem:torus_rgg_diam}, the diameter of the subgraph
  induced by a part $P$ of $\mathcal{P}^\ell$ has diameter as claimed
  with probability $1-n^{-C}$ for any constant $C>0$.  As there are
  only $4^\ell = n^α$ such parts, a union bound gives that all
  diameters fall within the claimed range with high probability.  To
  conclude, recall that w.h.p.\ every block with side length $x$
  induced by a square has $x^2 (1+o(1))$ vertices
  (\cref{lem:whp_parts_bound}).  Thus, the diameter of a block with
  $k$ vertices is in $Θ(k^{1/2}r^{-1})$.  Further with
  $r∈o(n^{1/2 - α/2})$ leaf blocks w.h.p.\ have $O(n^{1-α})$ vertices,
  by \cref{lem:partition_separators}.  Thus their diameter is in
  $Θ(n^{1/2 - α/2}r^{-1})$.
\end{proof}

\subsubsection{Fragmentation}
\label{sec:rgg-nice:rec:frag}

We show that the recursive partition $G[\mathcal{P}^{\ell}]$ has
bounded fragmentation.  Again, we start with the purely geometric
setting and prove that a disk does not intersect too many squares of a
grid.

\begin{lemma}\label{lem:num_square_cap_disk}
  Consider an axis aligned grid tiling of $ℝ²$ with squares of side
  length $s>0$ and a disk $D$ of radius $r>0$.  Then the number of
  squares intersected by $D$ is at most $O\parens*{\frac{r²}{s²}+1}$.
\end{lemma}
\begin{proof}
  Without loss of generality assume that $D$ is centered at the
  origin.  If a square
  $Q = [x_{1}, x_{1} + s] \times [y_{1}, y_{1}+1]$ intersects $D$,
  then it must be contained in the square
  $Q^{\star} = [-(r+s), r+s] \times [-(r+s), r+s]$.  The area of
  $Q^{\star}$ is $(2r+2s)²$, while each square only has area $s²$.
  This means that $\frac{(2r+2s)²}{s²} ∈ O\parens*{\frac{r²}{s²}+1}$
  non-intersecting squares of the tiling can lie inside $Q^{\star}$.
\end{proof}

We translate this to the setting of the recursive partition of a
random geometric graph and obtain the following theorem about its
fragmentation.

\begin{lemma}\label{lem:rgg:bounded_fragmentation}
  Let $G\sim\mathcal{G}(X, n, r)$ be a random geometric graph with
  ground space $X∈\{\mathcal{S}, \mathcal{T}\}$ and connection radius
  $r ∈ ω(\log^{3/4} n)$.  Further, for a constant $α∈(0,1)$ such that
  $r∈o(n^{1/2-α/2})$, let $\ell = α \log_{4}(n)$.  Then,
  asymptotically almost surely, the recursive partition
  $G[\mathcal{P}^\ell]$ has bounded-fragmentation.
\end{lemma}
\begin{proof}
  We consider a set of vertices $A$ that is contained in a ball of
  radius $k$ in $G$.  Further, let $c_1>1$ be a constant and denote by
  $\mathcal{B}_k$ the blocks of $\mathcal{P}^\ell$ with diameter
  between $k/c_1$ and $kc_1$.  We need to show that only a bounded
  number of blocks of $\mathcal{B}_k$ intersect $A$.


  We condition on the stretch event, which holds asymptotically almost
  surely (\cref{lem:stretch,lem:torus_stretch}).  Then, $A$ is
  contained in a geometric ball $B$ of radius $O(kr)$.  By
  \cref{lem:num_square_cap_disk}, $B$ intersects only
  $O\parens*{\frac{(kr)²}{s²}}$ squares of side length $s$ in any grid
  tiling of $ℝ²$.  This upper bound also applies to the number of
  blocks intersecting $B$ in each level of $\mathcal{P}^\ell$, as
  these can be extended into a tiling of $ℝ²$.  The blocks of
  $\mathcal{B}_k$ come from $O(\log c_1)$ many different levels in
  $\mathcal{P}^\ell$.  As they have diameter in $Θ(k)$, the geometric
  diameter and thus also the side length of these blocks is a.a.s.\ in
  $Θ(kr)$ (\cref{lem:partition_diam}).  Thus, the number of blocks of
  $\mathcal{B}_k$ intersecting $B$ is in
  $O\big(\frac{(kr)²}{(kr)²}\log c_1\big) = O(1)$, which concludes the
  proof.
\end{proof}

\subsection{Local Diametric Partners and Few Corners}
\label{sec:rgg-nice:corner}
In this section we show that \cref{item:prop-diam-partners} holds
a.a.s.\ on both square and torus RGGs while
\cref{item:prop-few-corners} additionally holds on square RGGs.  We
begin by considering the purely geometric setting and give a proof for
\cref{item:geom-diam-partners} on the unit square.  



\begin{figure*}
  \centering
  \begin{subfigure}[b]{0.3\textwidth}
      \includegraphics[page=1]{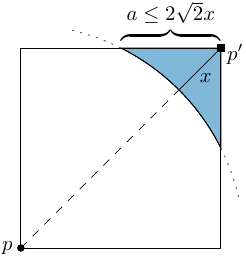}
      \subcaption{}
      \label{fig:geom_fringe_square:a}
  \end{subfigure}
  ~
  \begin{subfigure}[b]{0.3\textwidth}
      \includegraphics[page=2]{figs/square_geom_fringe.pdf}
      \subcaption{}
      \label{fig:geom_fringe_square:b}
  \end{subfigure}
  ~
  \begin{subfigure}[b]{0.3\textwidth}
      \includegraphics[page=3]{figs/square_geom_fringe.pdf}
      \subcaption{}
      \label{fig:geom_fringe_square:c}
  \end{subfigure}
  \caption{Part (a): visualization of point $P$ in corner of the unit
    square $S$, with the set of points $Q \subseteq S$ with distance
    at least $\sqrt{2} - x$ from $P$ shown in blue.  Part (b): right
    triangle used in the proof of \cref{lem:square_geom_fringe}.  Part
    (c): analogous situation on the torus, see also
    \cref{lem:torus_geom_fringe}.}
  \label{fig:geom_fringe_square}
\end{figure*}

\begin{lemma}\label{lem:square_geom_fringe}
  Let $S$ be the unit square and $p$ a point on $S$.  For every $x>0$,
  the set of points with distance at least $\sqrt{2} - x$ from $p$ can
  be covered by a disk of radius $O(x)$.
\end{lemma}
\begin{proof}
  Denote the set of points with distance at least $\sqrt{2} - x$ from
  $p$ by $Q$ and let $r = \sqrt{2} - x$.  Without loss of generality,
  we assume that $p$ is in the left and lower quadrant of $S$.  Note
  that moving $p$ towards the bottom left corner only increases $Q$
  inclusion-wise, so we can even assume that $p$ is located in the
  corner.  Then, $Q$ forms a region around the opposite corner $p'$ of
  $p$, see \cref{fig:geom_fringe_square:a}.  To show that $Q$ is
  contained in a disk of radius $O(x)$, we make a case distinction on
  $x$.  We first consider large $x$.  For $x ≥ \sqrt{2} - 1$, any
  point on the square $S$ has distance at most
  $\sqrt{2} ≤ x \cdot \frac{\sqrt{2}}{\sqrt{2}-1}∈O(x)$ from $p'$.
  
  Otherwise, we have $x < \sqrt{2} - 1$.  Then, the boundary of $Q$
  consists of two line segments and a circular arc, see also
  \cref{fig:geom_fringe_square:a}.  Denote the length of these line
  segments by $a$.  Every point $q∈Q$ is at distance at most $a$ from
  the corner $p'$ of $S$ opposite to $p$, so it remains to find an
  upper bound on $a$.  Applying the Pythagorean theorem (see
  \cref{fig:geom_fringe_square:b}), we obtain
  \[
    (1-a)² + 1² = (\sqrt{2} - x)²
  \]
  and thus
  \[
    a² - 2a = x² - 2 \sqrt{2} x.
  \]
  This quadratic equation has two solutions,
    \[
    a = 1 - \sqrt{x² - 2\sqrt{2}x + 1}
    \quad\text{or}\quad
    a = 1 + \sqrt{x² - 2\sqrt{2}x + 1}.
    \]
    We have $a ≤ 1$ as $S$ is a unit square, so only the first
    solution is relevant.  We derive an upper bound as follows.  Let
    $t = x² - 2\sqrt{2}x +1$.  Then
  \begin{align*}
    a & = 1 - \sqrt{t} 
      = \frac{
        (1 - \sqrt{t})
        (1 + \sqrt{t})
        }{
        1 + \sqrt{t}
        }
      = \frac{
        1 - t
        }{
        1 + \sqrt{t}
        }
      = \frac{
        x (2\sqrt{2} - x)
        }{
        1 + \sqrt{t}
        }
       ≤ 2 \sqrt{2} x.
  \end{align*}
  To summarize, either $x ≥ \sqrt{2} - 1$ and any point on the square
  $S$ has distance at most $O(x)$ from $p'$, or $x < \sqrt{2} - 1$ and
  by the derivation above, the distance between $Q$ and $p'$ is at
  most $2 \sqrt{2} x ∈ O(x)$.  In both cases, $Q$ is contained in a
  disk of radius $O(x)$.
\end{proof}

The same argument works analogously on the torus, yielding the
following.

\begin{lemma}\label{lem:torus_geom_fringe}
  Let $T$ be the unit torus and $P∈T$ a point.  For $x > 0$, the set
  of points with distance at least $\frac{\sqrt{2}}{2} - x$ from $P$ is
  contained in a disk of radius $O(x)$
\end{lemma}
\begin{proof}
  On the torus, the set of points with distance $\frac{ \sqrt{2} }{2}$
  is shaped like four mirrored and scaled down copies of the analogous
  set on a square, see also \cref{fig:geom_fringe_square:c}.  The
  statement thus follows from \cref{lem:square_geom_fringe}.
\end{proof}

We can scale the distances considered in \cref{lem:square_geom_fringe}
by a factor of $\sqrt{n}$ and obtain statements about the geometric
ground space of square and torus RGGs.  We get that for any vertex $v$
and $x>0$, the set of points with distance from $v$ at least the
geometric diameter minus $x$ is contained in a geometric disk of
radius $O(x)$.  By the results of \cref{sec:rgg-nice:stretch} the
graph distance of vertices with geometric distance $d$ is between
$\ceil*{\frac{d}{r}}$ and $\ceil*{\frac{d}{r}}(1+O(r^{-4/3}))$.
Together, this allows us to show the following.

\begin{lemma}\label{thm:concentrated_fringes}
  Asymptotically almost surely, a square or torus random geometric
  graph $G\sim\mathcal{G}(X, n, r)$ with
  $X∈\{\mathcal{S}, \mathcal{T}\}$ and connection radius
  $r∈ω(\log^{3/4} n)$ has $d_{\mathrm{local}}$-local diametric
  partners with $d_{\mathrm{local}}∈O(n^{1/2}r^{-7/3} + 1)$.
\end{lemma}
\begin{proof}
  Let $v∈V(G)$ be a vertex and for $x∈ℕ$ let $w$ be a $x$-diametric
  partner of $v$, i.e., $\dist_G(v,w) \ge \diam_G - x$.  If
  $x∈Ω(\diam_G)$, all vertices $V(G)$ have distance $O(x)$ from $v$
  and are thus contained in a ball of radius $O(x)$.  Otherwise,
  $d_G(v,w) ∈Ω(diam_G)$.  By \cref{lem:rgg_diam,lem:torus_rgg_diam},
  the diameter of $G$ is at least $\diam_G \ge \frac{D}{r} - 1$, where
  $D∈Θ(\sqrt{n})$ is the geometric diameter of the ground space of
  $G$.  This means that a.a.s.\ $\dist_G(v,w) ∈ Ω(n^{1/2}r^{-1})$ and
  in particular the geometric distance $d$ between $v$ and $w$ is at
  least $r \log n$.  Thus by \cref{lem:stretch,lem:torus_stretch} we
  have $d_G(v,w) \le \frac{d}{r}(1+s)$ for $s∈O(r^{-4/3})$.  We thus
  get
  \begin{align*}
    \frac{d}{r}(1+s) & \ge \frac{D}{r} -1 -x \\
    d(1+s) & \ge D - r - xr \\
    d & \ge (D - r - xr) \frac{1}{1+s}.
  \end{align*}
  With $s_G≥0$ we have $\frac{1}{1+s_G} \ge 1-s_G$.  Hence,
  \begin{align*}
    d \ge (D - r - xr) (1-s) & \ge D - r - xr - D s + sr + sxr\\
    & \ge D - r - xr - D s.
  \end{align*}

  This means that the vertex $w$ has distance from $v$ at least the
  geometric diameter of $\mathcal{S}$, respectively $\mathcal{T}$,
  minus $(r + x r + D s)$.  By \cref{lem:square_geom_fringe},
  respectively \cref{lem:torus_geom_fringe}, points at this distance
  from $v$ are contained in a geometric disk of radius
  $O(r + x r + D s) = O(r + xr + n^{1/2}r^{-4/3})$.  By
  \cref{lem:rgg_region_not_empty}, $G$ w.h.p.  contains a vertex $v_c$
  within distance $O(\sqrt{\log n})$ of the center of this disk.  Any
  vertex within the disk then has graph distance at most
  $O(1 + x + n^{1/2}r^{-4/3})$ from $v_c$.  This means that all
  $x$-diametric partners of $v$ are contained in a $G$-ball of radius
  $O(x + n^{1/2}r^{-4/3} + 1)$.
\end{proof}

We also show that square RGGs have few corners
(\cref{item:prop-few-corners}).

\begin{lemma}\label{thm:rgg_few_corners}
  There exists $d_{\mathrm{corner}}∈Θ(n^{1/2-7/3r}+1)$ such that a
  square random geometric graph $G\sim\mathcal{G}(\mathcal{S},n,r)$
  with connection radius $r∈ω(\log^{3/4} n)$ has
  $d_{\mathrm{corner}}$-few corners, asymptotically almost surely.
\end{lemma}
\begin{proof}
  We consider four small regions around the corners of $\mathcal{S}$
  that we call \emph{corner squares}.  We show that for $x\ge 0$ all
  vertices outside these regions have no $x$-diametric partners.
  Then, by contraposition any vertex with at least one $x$-diametric
  partner lies in a corner square.  Like before, we show this by first
  making a purely geometric argument and then applying the stretch
  bounds.  Afterwards, it remains to show that each corner square can
  be covered by a ball of radius $O(x + d_{\mathrm{corner}})$.

  \begin{figure}
    \centering
    \includegraphics{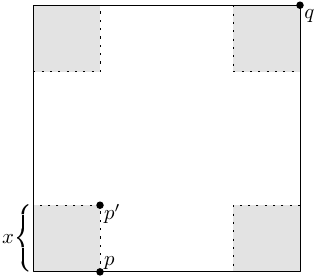}
    \caption{Visualization of square $\mathcal{S}$ with the four
      corner squares in gray; the points $p$ and $q$
      maximize the distance between a point in the lower left quadrant
      and any other point and an upper bound for their distance is
      easily found by considering the detour via $p'$.}
    \label{fig:few_corner_square}
  \end{figure}

  For the geometric argument, let $\ell>0$ be a length parameter to be
  determined later.  We define the corner squares with side length $\ell$
  as follows, see also \cref{fig:few_corner_square}.  Let
  $c_1, \dots, c_4$ be the four corners of $\mathcal{S}$.  Then, for
  $i∈[4]$ we define the corner square $C_i$ as the subset of
  $\mathcal{S}$ that lies within the axis aligned square of side
  length $2\ell$ and center $c_i$, i.e.,
  \[
    C_i = \{p ∈ \mathcal{S} \mid ∃ d_x, d_y \text{ with
    } -\ell < d_x,d_y < \ell \text{ such that } p = c_i + (d_x,d_y)\}.
  \]

  Without loss of generality let $C$ be the corner square in the
  bottom left corner.  We consider a point $p∈\mathcal{S} \setminus $
  and give an upper bound for the maximum distance from $p$ to any
  other point $q ∈ \mathcal{S}$.  We can pessimistically assume that
  $p$ lies at $(0,\ell)$ and $q$ lies at the top right corner of
  $\mathcal{S}$, see also \cref{fig:few_corner_square}.  Choosing
  $p' = (\ell,\ell)$ as the top right corner of $C$, we have
  \[
    d_{\mathbb{E}}(p,q) < d_{\mathbb{E}}(p,p') + d_{\mathbb{E}}(p',q) = \ell + \sqrt{2n} - \sqrt{2}\ell < \sqrt{2n} - 0.4\ell,
  \]
  which concludes the geometric argument.

  We condition on the stretch event of \cref{lem:stretch}, which holds
  asymptotically almost surely.  This means that for any pair of
  vertices with geometric distance at least $d \ge r \log n$, the
  graph distance is at most $\ceil*{d/r \cdot (1+s)}$ for some
  $s∈O(r^{-4/3})$.
  Thus for any vertex
  $v∈ V(G) \cap \mathcal{S} \setminus (C_1∪C_2∪C_3∪C_4)$ located
  outside the corner squares and any other vertex $w∈V(G)$ we have
  \begin{align*}
    d_G(v,w) & \le \ceil*{\frac{d_{\mathbb{E}}(v,w)}{r}(1+s)}\\
    & < \frac{d_{\mathbb{E}}(p,q)}{r}(1+s)+1 \\
    & < \frac{\sqrt{2n} - 0.4 \ell}{r} (1+s) + 1 \\
    & < \frac{\sqrt{2n} + \sqrt{2n}s - 0.4 \ell}{r} + 1.
  \end{align*}
  By \cref{lem:rgg_diam}, we have w.h.p.\
  $\diam_G ≥ \frac{\sqrt{2n}}{r} - 1$.  Setting
  $\ell = \frac{\sqrt{2n}s + xr + 2r}{0.4} $, we thus have
  \[
    d_G(v,w) < \frac{\sqrt{2n}}{r} - x -1 \le \diam_G - x.
  \]
  In other words, no vertex of $G$ has distance $\diam_G - x$ or more
  from $v$, hence $v$ has no $x$-diametric partner.  Note that
  we have $\ell ∈ O(xr + n^{1/2}s + r) = O(xr + n^{1/2-4/3r} + r)$. 

  This means that for some
  $d_{\mathrm{corner}}∈Θ(n^{1/2-7/3r}+1) \subseteq Ω(\sqrt{n}sr^{-1})$
  and any $x ≥ 0$ we can choose $\ell∈O(xr + d_{\mathrm{corner}} r)$
  such that any vertex outside the corner squares of side length
  $\ell$ does not have $x$-diametric partners.
  This means that any vertex with at least one $x$-diametric partner
  lies inside one of four squares of side length $\ell$.  It remains to
  show that each of these squares can be covered by $O(1)$ $G$-balls
  of radius $O(x + d_{\mathrm{corner}})$.  To this end, consider
  without loss of generality the corner square $C$ in the bottom left
  corner.  By \cref{lem:rgg_region_not_empty}, $G$ w.h.p.\ contains a
  vertex $v_C$ within geometric distance $O(\sqrt{\log n})$ of the
  geometric center of $C$.  Any other vertex $w_C ∈ V(G) \cap C$ then
  has geometric distance $O(\ell) = O(xr + d_{\mathrm{corner}}r)$ and
  thus also graph distance at most
  $O(\ell/r) = O(x + d_{\mathrm{corner}})$ from $v_C$.  Thus,
  $V(G) \cap C$ is contained in the closed
  $O(x + d_{\mathrm{corner}})$ neighborhood of some vertex of $G$.

  To conclude, we have shown that there is a
  $d_{\mathrm{corner}}∈Θ(n^{1/2}r^{-7/3} + 1)$ and such that for any
  $x ≥ 0$ the set of vertices with at least one $x$-diametric partner
  can be covered by a constant number of $G$-balls of radius
  $O(x + d_{\mathrm{corner}})$.
\end{proof}

\subsection{Computing the Diameter}%
\label{sec:rgg:diameter}

We are now ready to apply the diameter algorithm from \cref{sec:algo}.
For a better overview we first give a summary of the properties we
have shown above.  Let $G$ be a random geometric graph with connection
radius $r = n^ρ$ for constant $ρ∈(0,1)$.  Then, asymptotically almost
surely $G$ has $d_{\mathrm{local}}$-local diametric partners for
$d_{\mathrm{local}} ∈ O(n^{\frac12 - \frac73ρ}+1)$ (see
\cref{thm:concentrated_fringes}).  Additionally, if $G$ has a square
ground space, it a.a.s.\ has $d_{\mathrm{corner}}$-few corners with
$d_{\mathrm{corner}}∈Θ(n^{\frac12-\frac73ρ}+1)$ (see
\cref{thm:rgg_few_corners}).

Moreover, for a constant $x∈(0,1)$ with $x > 2ρ$, let
$\ell = (1-x) \log_4 n$.  Then, the recursive partition
$G[\mathcal{P}^\ell]$ a.a.s.\ is balanced
(\cref{lem:partition_separators}) and has $(\frac12,ρ)$-small separators
(see \cref{lem:partition_separators}), size-dependent diameters (see
\cref{lem:partition_diam}), and bounded fragmentation (see
\cref{lem:rgg:bounded_fragmentation}).  Moreover, the leaf blocks of
$G[\mathcal{P}^\ell]$ have $Θ(n^{x})$ vertices and each block of size
$k$ has diameter $Θ(k^{\frac12}n^{-ρ})$ (see \cref{lem:partition_diam})
and thus leaf blocks have diameter $Θ(n^ε)$ for some constant $ε>0$.

To apply \cref{thm:properties_algo}, we additionally need an upper bound on
the degeneracy of $G$.  This is easily obtained using the concentration
bounds on the number of vertices inside a sufficiently large region.

\begin{lemma}\label{lem:rgg_maxdeg}
  A random geometric graph $G\sim\mathcal{G}(X, n, r)$ with ground
  space $X∈\{\mathcal{S}, \mathcal{T}\}$ and connection radius
  $r = n^x$ for constant $x∈(0,1)$ has an expected average degree in
  $O(n^{2x})$ and a maximum degree in $O(n^{2x})$ with high
  probability.
\end{lemma}
\begin{proof}
  Let $v∈V(G)$ be a vertex.  Then the degree of $v$ is equal to the
  number of vertices falling into a region of radius $r$ and hence
  area $A_x ∈ O(n^{2x})$ around $v$, i.e., the expected average degree
  is in $O(n^{2x})$.  With \cref{lem:area_vertex_count}, this also
  means that the degree of $v$ is at most
  $A_x \cdot \parens*{1+ \sqrt{\frac{3c \log n}{A_x}}}$ with
  probability at least $1 - O(n^{-c})$ for any constant $c$.  For a
  sufficiently high constant $c$, a union bound over all vertices
  shows that all vertices have degree in $O(n^{2x})$ with high
  probability.
\end{proof}


Applying \cref{thm:properties_algo}, we thus get the following running
times, depending on the exponent of the connection radius $ρ$.

\begin{lemma}\label{thm:rgg_running_time}
  Let $G$ be a torus or a square random geometric graph with
  connection radius $r=n^ρ$ for constant $ρ∈(0,\frac12)$.  Then, $G$
  admits a recursive partition $\mathcal{P}$ such that asymptotically
  almost surely the algorithm from \cref{thm:properties_algo} computes
  the diameter of $G$ in time
  $\tilde{O}\parens*{n^{\max({\frac32+3ρ}, 2-\frac{2}{3}ρ)}}$.  If the
  ground space of $G$ is the square $\mathcal{S}$, the running time is
  in $\tilde{O}\parens*{n^{\max(\frac32+3ρ,2-\frac{10}{3}ρ)}}$.
%
\end{lemma}
\begin{proof}
  Asymptotically almost surely, we can rely on the properties
  summarized above.  The running time guarantee from
  \cref{thm:properties_algo} depends on a parameter $k$ such that the
  recursive partition contains a flat partition with blocks of size
  $k$ and diameter in $Ω(d_{\mathrm{local}})$ and optionally in
  $Ω(d_{\mathrm{corner}})$.  We have
  $Ω(d_{\mathrm{local}}) = Ω(d_{\mathrm{corner}}) = Ω(n^{\frac12 - \frac73ρ} + 1)$.
  This means that there is a phase transition at $ρ = \frac{3}{14}$,
  above which $d_{\mathrm{local}}$ and $d_{\mathrm{corner}}$ are no
  longer growing in $n$.  We thus consider the two cases
  $ρ < \frac{3}{14}$ and $ρ ≥ \frac{3}{14}$ separately.

  For the first case, we use the recursive partition
  $G[\mathcal{P}^\ell]$ with $\ell = (1-x) \log_4 n$ for
  $x = \frac{6}{14}$, i.e., leaves of $\mathcal{P}$ have size
  $n^{\frac{6}{14}}$ and diameter $n^{\frac{3}{14}-ρ}$.  Then, for
  $k = n^{1-\frac83 ρ}$ blocks of size $k$ have diameter
  $n^{\frac12 - \frac73 ρ}∈Ω(d_{\mathrm{local}}) = Ω(d_{\mathrm{corner}})$.
  Also, we have
  $1-\frac{8}{3} ρ > 1- \frac{8}{3} \cdot \frac{3}{14} = \frac{6}{14}$,
  and thus $k∈Ω(n^x)$, i.e., $\mathcal{P}$ contains flat partitions
  of size $Θ(k)$.

  With this choice for the parameters, the running time for torus RGGs given by
  \cref{thm:properties_algo} is in
  \begin{align*}
    \tilde{O}\parens*{
    n^{1+α+β} \cdot d +
    \min\braces*{
    nkd +
    k^{2α}n^{1+2β} +
    k^{2α-1}n^{1+α+3β},\;
    kn^{1+α+β}
    }
    },
  \end{align*}
  with $α=\frac12$, $β=ρ$, $d=n^{2ρ}$.  Considering each term separately, we have
  \begin{align*}
    n^{1+α+β} \cdot d & = n^{1+\frac12+ρ}\cdot n^{2ρ} = n^{\frac32 + 3ρ},\\
    nkd & = n\cdot n^{1-\frac83ρ} \cdot n^{2ρ} = n^{2-\frac23ρ},\\
    k^{2α}n^{1+2β} & = n^{1-\frac83ρ} \cdot n^{1+2ρ} = n^{2-\frac23ρ},\\
    n^{1+α+3β}k^{2α-1} & = n^{1+\frac12+3ρ} k^{0} = n^{\frac{3}{2}+3ρ},\\
    kn^{1+α+β} & = n^{1-\frac83ρ} \cdot n^{1+\frac12 + ρ} = n^{\frac52 - \frac53ρ}.
  \end{align*}
  For $0<ρ<\frac{3}{14}$ the term $kn^{1+α+β} = n^{\frac52 - \frac53ρ}$ is always
  larger than the other terms in the minimum.  This means that the
  running time is in
  \[
    \tilde{O}\parens*{n^{\max{(\frac32+3ρ}, 2-\frac{2}{3}ρ)}}.
  \]
  
  \noindent{}Square RGGs additionally have $d_{\mathrm{corner}}$-few
  corners, so \cref{thm:properties_algo} gives a running time in
  \begin{align*}
    \tilde{O}\parens*{
    n^{1+α+β}\cdot d +
    \min\braces*{
    k²d +
    k^{1+2α}n^{2β}
    + k^{2α}n^{α+3β},\;
    k²n^{α+β}
    }
    },
  \end{align*}
  with $α=\frac12$, $β=ρ$, $d=n^{2ρ}$, and $k = n^{1-\frac83ρ}$.  We
  again consider each term separately.  We have
  \begin{align*}
    n^{1+α+β} \cdot d & = 
                        n^{\frac32 + 3ρ},\\
    k²d & = n^{2(1-\frac83ρ)}d = n^{2 - \frac{16}{3}ρ}d = n^{2-\frac{10}{3}ρ},\\
    k^{1+2α}n^{2β} & = n^{2(1-\frac83ρ)} \cdot n^{2ρ} = n^{2-\frac{10}{3}ρ},\\
    n^{α+3β}k^{2α} & = n^{\frac12 + 3ρ} \cdot n^{1-\frac83ρ} = n^{\frac32 + \frac13ρ},\\
    k²n^{α+β} & = n^{2(1-\frac83ρ)} \cdot n^{\frac12+ρ} = n^{\frac52-\frac{13}{3}ρ}.
  \end{align*}
  Similar to the torus case, for $0 < ρ < \frac{3}{14}$ the term
  $n^{\frac52 - \frac53ρ}$ is always larger than the other terms in
  the minimum.  This means that the running time is in
  \[
    \tilde{O}\parens*{n^{\max(2-\frac{ 10 }{3}ρ,\frac32+3ρ)}}.
  \]

  For the case $ρ≥\frac{3}{14}$, let $ε < \frac12 - ρ$ be a positive
  constant\footnote{For the positivity, recall that we generally
    assume $ρ<\frac12$, see \cref{sec:rgg-nice:prelim}.}.  Then,
  $ρ+ε < \frac12$ and $2(ρ+ε) < 1$.  We use a recursive partition
  $G[P^\ell]$ with leaf size $n^{2ρ+\frac{ε}{2}}$ and subgraph sizes
  $k = n^{2ρ+ε}$.  Then, the diameter of size $k$ blocks is
  $d_k ∈ Θ(n^{\frac{2ρ+ε}{2} - ρ}) = Θ(n^{ε/2})$ and thus we have
  $d_k ∈ Ω(\min\{d_{\mathrm{local}},d_{\mathrm{corner}}\}) = Ω(1)$.
  Again, considering each term of the running time separately, we have
  \begin{align*}
    n^{1+α+β} \cdot d & 
                      = n^{\frac32 + 3ρ},\\
    nkd & = n\cdot n^{2ρ+ε} \cdot n^{2ρ} = n^{1+4ρ+ε},\\
    k^{2α}n^{1+2β} & = n^{2ρ+ε} \cdot n^{1+2ρ} = n^{1+4ρ+ε},\\
    n^{1+α+3β}k^{2α-1} & = n^{1+\frac12+3ρ} k^{0} = n^{\frac{3}{2}+3ρ},\\
    k^2n^{α+β} &= n^{2(2ρ+ε)}\cdot n^{\frac12+ρ} = n^{\frac12 + 5ρ + 2ε}
  \end{align*}
  for torus RGGs.
  However, we have $ρ+ε<\frac12$ thus
  $n^{\frac32+3ρ}$ dominates $n^{1+4ρ+ε}$ and $n^{\frac12+5ρ+2ε}$.
  Together with the running time analysis for the case
  $ρ<\frac{3}{14}$, this means that for any value of $ρ∈(0,\frac12)$,
  the running time on torus RGGs is in
  \[
    \tilde{O}\parens*{n^{\max{(\frac32+3ρ}, 2-\frac{2}{3}ρ)}}.
  \] 

  For square RGGs we have
  \begin{align*}
    n^{1+α+β} \cdot d & = n^{1+\frac12+ρ}\cdot n^{2ρ} = n^{\frac32 + 3ρ},\\
    k²d & = n^{2(2ρ+ε)}d = n^{4ρ+2ε}n^{2ρ} = n^{6ρ+2ε}, \\
    k^{1+2α}n^{2β} & = k^2n^{2ρ} = n^{6ρ+2ε}, \\
    n^{α+3β}k^{2α} & = n^{\frac12 + 3ρ} \cdot n^{2ρ+ε} = n^{\frac12+5ρ+ε}, \\
    k^2n^{α+β}    & = n^{\frac12 + 5ρ + 2ε}.
  \end{align*}
  Again, with $ρ+ε<0.5$ we have $\frac12+5ρ+ε < \frac32+3ρ$ and
  $\frac12+5ρ+2ε < \frac32+3ρ$.  Together with the running time analysis
  for the case $ρ<\frac{3}{14}$, this means that for any value of
  $ρ∈(0,\frac12)$, the running time on square RGGs is in
  \[
    \tilde{O}\parens*{n^{\max{(\frac32+3ρ}, 2-\frac{10}{3}ρ)}}.\qedhere
  \]
\end{proof}

Equivalently the running times can be written as
$\tilde{O}\big(n^{\max({\frac12+ρ}, 1-\frac{8}{3}ρ)}m\big)$ (torus)
and $\tilde{O}\big(n^{\max(\frac12+ρ,1-\frac{16}{3}ρ)}m\big)$
(square).  The following theorem follows directly, as RGGs with
connection radius $r=n^ρ$ have expected average degree $Θ(n^{2ρ})$
(see \cref{lem:rgg_maxdeg}).

\runningTimeRGG*

\subsection{Analysis of iFUB}\label{sec:ifub}
In this section we rely on the stretch bounds and the concentration of
the vertices to analyze the running time of the iFUB algorithm on
random geometric graphs.  We consider iFUB with the \emph{2-sweep heuristic},
i.e., the algorithm chooses a central vertex as follows.  First, the
algorithm performs a BFS from an arbitrary vertex $v$ and picks a
vertex $w$ in the last layer, i.e., with maximum distance from $v$.
Then, a second BFS is performed from $w$ and the vertex $c$ is chosen
half the way on a shortest path between $w$ and a vertex with maximum
distance from $w$.

In the following, we begin by showing that the vertex $w$ selected by
the first BFS is likely located close to a corner of the square
ground space.

\subparagraph*{Analysis of 2-sweep.}  We begin with a geometric
argument showing that for any point there is a corner that is further
away than a second point not close to any corner.  See also
\cref{fig:corner_distance} for a visualization.

\begin{lemma}\label{lem:square_corner_dist}
  Let $S$ be a square and let $p$ and $q$ be points on $S$, such that
  $q$ has distance at least $x$ from every corner of $S$.  Then, there
  is a corner $c^*$ of $S$ such that
  $d_{\mathbb{E}}(p, q) + 0.23 x \le d_\mathbb{E}(p,c^{*})$, i.e., the
  distance from $p$ to $c^*$ is at least $0.23 x$ longer than the
  distance from $p$ to $q$.
\end{lemma}
\begin{figure*}
  \centering
  \includegraphics{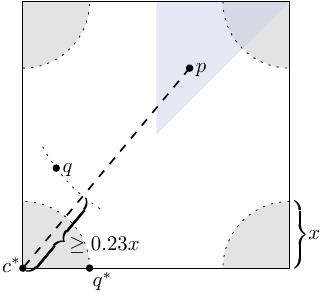}
  \caption{Visualization of the setting in
    \cref{lem:square_corner_dist}. Without loss of generality we
    assume that $p$ is located in the blue triangle.  If $q$ lies far
    from every corner, it is closer to $p$ than the furthest corner
    $c^*$, additionally $q^*$ marks the position of $q$ that maximizes
    the distance to $p$.  Without loss of generality we assume that
    $p$ is located in the blue triangle.}
  \label{fig:corner_distance}
\end{figure*}
\begin{proof}
  Without loss of generality we assume that $S$ is the unit square
  $[0,1)\times [0,1)$.  Further, we can assume that $p = (p_{x}, p_{y})$ lies in
  the upper left diagonal half of the upper right quadrant of $S$,
  i.e., $p_{x} \ge \frac12$ and $p_{y}\ge p_{x}$.  See also
  \cref{fig:corner_distance}.  We choose $c^{*}$ as the bottom left
  corner.

  The permissive region for $q$ is obtained from $S$ by removing
  quarter circles of radius $x$ centered at the corners of $S$.  We
  first consider the case $x \le \frac12$.  Then, with $p$ in the
  upper left diagonal half of the upper right quadrant of $S$, the
  furthest position $q^*$ of $q$ is at the intersection of the bottom
  side of $S$ and the bottom left quarter circle.

  We have $|pq| \le |pq^*| = \sqrt{(p_{x}-x)²+p_{y}²}$ and
  $|pc^{*}| = \sqrt{p_{x}^2+p_{y}^2}$.  Thus we have
  \[
    |pc^{*}| - |pq| \ge \sqrt{p_{x}^2+p_{y}^2} - \sqrt{(p_x-x)²+p_{y}²}.
  \]
  This difference is increasing in $p_x$ and decreasing in $p_y$.
  Therefore, it is minimized at $p=(\frac12,1)$, giving us
  $|pc^{*}| - |pq| \ge \frac{\sqrt{5}}{2} - \sqrt{(\frac12-x)^2+1}$.

  We consider
  $\big(\frac{\sqrt{5}}{2} - \sqrt{(\frac12-x)^2+1}\big) / x$ and find
  that it is decreasing in $x$.  This means that the expression has
  its minimum of $\sqrt{5}-2$ at $x=\frac{1}{2}$.  We have thus shown
  \[
    (|pc^{*}| - |pq^{*}|)/x \ge \sqrt{5}-2 > 0.23,
  \]
  which implies $|pq| + 0.23x \le |pc^*|$ as claimed.

  For $x>0.5$ the worst-case position of $q$ is at
  $q^*=(\frac12, \sqrt{x²-\frac{1}{2}^{2}})$ and the worst-case
  position of $p$ is still $(\frac12,1)$.  We have
  \[
    \frac{|pc^{*}| - |pq^{*}|}{x} \ge
    \frac{\frac{\sqrt{5}}{2} - \parens*{1-\sqrt{x²-\frac14}}}{x} =
    \frac{\frac{\sqrt{5}}{2} - 1 + \sqrt{x²-\frac14}}{x}.
  \]
  This is at least $\sqrt{5}-2$ for all $x>\frac12$, which concludes
  the proof.
\end{proof}

We apply this to show that in a square RGG the furthest neighbor of
every vertex lies close to a corner of the square.

\begin{lemma}
  \label{lem:furthest_vertex_in_corner}
  Let $G\sim\mathcal{G}(\mathcal{S}, n, r)$ be a square random
  geometric graph with $r∈ω(\log^{3/4} n)$.  Consider a vertex
  $v∈V(G)\cap S$ and a maximally distant $w∈N(v,\ecc(v))$ of $v$.
  Then, asymptotically almost surely, the geometric distance of $w$ to
  some corner of $\mathcal{S}$ is in
  $O(n^{\frac12}r^{-\frac43} + \sqrt{\log n})$.
\end{lemma}
\begin{proof}
  Without loss of generality, we assume $v$ to be in the upper right
  quadrant of $\mathcal{S}$ and we show that $w∈N(v,\ecc(v))$ has
  geometric distance in $O(n^{\frac12}r^{-\frac43})$ from the lower left
  corner $c^{*}$ of $\mathcal{S}$ located at the origin.

  We condition on the stretch event of \cref{lem:stretch}, which holds
  asymptotically almost surely.  Then every pair of vertices with
  geometric distance at least $d\ge r \log n$ has graph distance at
  most \( \ceil*{\frac{d}{r}\cdot(1+s)} \) for some $s∈O(r^{-\frac43})$.

  Assume towards a contradiction that $w$ has distance more than
  $ε = 5 \sqrt{2n} s + 5\sqrt{\log n}$ from all corners of
  $\mathcal{S}$.  We show that then there is another vertex with higher
  graph distance than $w$ from $v$.  By
  \cref{lem:rgg_region_not_empty}, there is a vertex $w'$ in $G$ with
  distance from the origin at most $\sqrt{\log n}$ asymptotically
  almost surely.

  It remains to show that $w'$ has higher distance from $v$ than $w$
  in the graph.  By \cref{lem:square_corner_dist} we have
  \begin{align*}
    d_\mathbb{E}(v, w) &\le d_\mathbb{E}(v, c^{*}) - 0.23 ε.
  \end{align*}
  Further, by the triangle inequality,
  \[
    d_\mathbb{E}(v, c^{*}) \le d_\mathbb{E}(v, w') + d_\mathbb{E}(w', c^{*}) ≤ d_\mathbb{E}(v, w') + \sqrt{\log n}.
  \]
  Combining these two inequalities we derive
  \begin{align*}
    d_\mathbb{E}(v, w) & \le d_E(v,w') + \sqrt{\log n} - 0.23 ε \\
                       & = d_\mathbb{E}(v, w') + \sqrt{\log n} - 0.23 \cdot (5 \sqrt{2n}s + 5\sqrt{\log n}) \\
    & < d_\mathbb{E}(v, w') - 1.15\, \sqrt{2n}s.
                         \intertext{As
      $d_\mathbb{E}(v,w')$ is at most the diagonal of $\mathbb{S}$,
      $\sqrt{2n}$, we further have}
    d_\mathbb{E}(v, w) & < d_\mathbb{E}(v, w') - 1.15 \, d_\mathbb{E}(v,w') \, s
    = d_\mathbb{E}(v, w') (1 - 1.15 \, s).
  \end{align*}

  We now compare the graph theoretic distance from $v$ to $w$ and to
  $w'$.  By~\cref{eq:stretch-trivial} we get a lower bound
  \[
    \dist_{G}(v, w') \ge \ceil*{\frac{d_\mathbb{E}(v, w')}{r}}.
  \]
  Applying the stretch bounds, we further get
  \begin{align*}
    \dist_{G}(v, w) & \le \ceil*{ \frac{d_\mathbb{E}(v,w)}{r} (1+s) }\\
                    & \le \ceil*{\frac{d_\mathbb{E}(v,w')}{r} (1-1.15s) (1+s)}\\
                    & < \ceil*{\frac{d_\mathbb{E}(v,w')}{r}}.
  \end{align*}
  This means that $\dist_G(v,w') > \dist_G(v,w)$, contradicting the
  assumption that no other vertex has higher distance from $v$ as $w$.
  We conclude that asymptotically almost surely for every vertex $v$,
  every vertex $w∈N(v,\ecc(v))$ has distance at most
  $ε=5 \sqrt{2n} s + 5 \sqrt{\log n}$ from $c^*$.
\end{proof}

This means that a 2-sweep gives a good lower bound for the diameter of
square RGGs.  Additionally, we show that on square RGGs a central
vertex chosen this way is located close to the geometric center of the
square.

\begin{lemma}\label{lem:ifub:2sweep}
  Let $G\sim\mathcal{G}(\mathcal{S}, n, r)$ be a square random geometric
  graph with $r∈ω(\log^{3/4} n)$ and let $v_c$ be the central vertex
  chosen after a 2-sweep.  Then, asymptotically almost surely, $v_c$
  lies within a geometric distance of
  $O(n^{\frac12}r^{-\frac23} + \sqrt{\log n})$ from the geometric center of
  $\mathcal{S}$.
\end{lemma}
\begin{proof}
  Let $v∈V(G)$ be an arbitrary starting vertex, let $w∈N_G(v,\ecc(v))$
  be a maximally distant vertex from $v$, and let $w'∈N_G(w, \ecc(w))$
  be maximally distant from $w$.  Further let $v_c$ with
  $|\dist_G(w,v_c) - \dist_G(w',v_c)| \le 1$ and
  $\dist_G(w,v_c) + \dist_G(w',v_c) = \dist_G(w,w')$ the central
  vertex chosen with the 2-sweep.  We condition on the stretch event,
  i.e., in the following for every pair of vertices $u$, $v$ with
  geometric distance $d$ in $ω(r \log n)$ we can assume
  $d_G(u,v) \le \frac{d}{r}(1+s)$, with $s∈O(r^{-\frac43})$.  By
  \cref{lem:furthest_vertex_in_corner}, $w$ and $w'$ each lie within
  geometric distance of $O(n^{\frac12}r^{-\frac43} + \sqrt{\log n})$ from some
  corner of $\mathcal{S}$.  As the opposite corners have geometric
  distance $\sqrt{2n}$ and all other corners have distance at most
  $\sqrt{n}$ it follows that $w$ and $w'$ lie within geometric
  distance of $O(n^{\frac12}r^{-\frac43} + \sqrt{\log n})$ from opposite
  corners of $\mathcal{S}$.  With the approximate location of $w$ and
  $w'$ known it remains to locate $v_c$.

  Denote the Euclidean midpoint of the segment $ww'$ by $m$ and its
  length $d_{\mathbb{E}}(w,w')$ as $D$.  For $x∈\{w,w'\}$, we have
  $\dist_G(x, v_c) \le \ceil*{\frac{dist_G(w,w')}{2}} \le \frac{D}{2r}(1+s)+1$.
  Using \cref{eq:stretch-trivial} this implies
  $\dist_\mathbb{E}(x,v_c) \le \frac{D}{2}(1+s)+2 = \frac{D}{2} + O(Dr^{-\frac43})$.
  Thus $v_c$ lies in the lens formed by the intersection of circles of
  radius $\frac{D}{2} + O(Dr^{-\frac43})$ centered at $w$ and $w'$.  Along
  the line through $w$ and $w'$, the lens has length
  $O(Dr^{-4/3})$.  To bound the width $h$ in the orthogonal direction,
  we need the distance from $m$ to either of intersection point of the
  circles, see also \cref{fig:ifub_2sweep_triangle}.
  \begin{figure}
    \centering
    \includegraphics{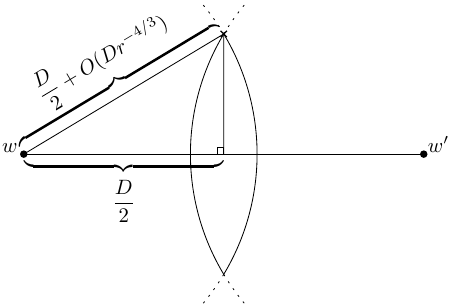}
    \caption{Visualization of lens and right triangle in proof of
      \cref{lem:ifub:2sweep}.}
    \label{fig:ifub_2sweep_triangle}
  \end{figure}
  Observe that $x$, $m$, and either intersection point form a right
  triangle.  By the Pythagorean theorem,
  \[
    \parens*{\frac{D}{2} + O(Dr^{-\frac43})}^2 = \parens*{\frac{D}{2}}^2 + h².
  \]
  To bound $h$ from above, set $a = \frac{D}{2}$ and
  $b = O(Dr^{-\frac43})$.  Then
  \begin{align*}
    h &=
        \sqrt{(a+b)²-a²} =
        \sqrt{2ab + b²} =
        \sqrt{2ab \Big(1+\frac{b}{2a}\Big)} =
        \sqrt{2ab} \sqrt{1+\frac{b}{2a}},
        \intertext{
        and thus, using that $\sqrt{1+x}\le 1+\frac x 2$ holds
        for $x>0$,}
        h &\le
        \sqrt{2ab} \parens*{1+\frac{b}{4a}}.
        \intertext{
        We have $\sqrt{2ab}∈O(D\cdot r^{-\frac23})$,
        $\frac{b}{4a}∈O(r^{-\frac43})$ and hence}
        h &\le D \cdot r^{-\frac23} (1+O(r^{-\frac43}))
  \end{align*}
  With $D∈O(\sqrt{n})$, it follows that $v_c$ lies within a Euclidean
  distance of $O(D \cdot r^{-\frac23})$ from $m$.  As $w$ and $w'$ are
  within $O(n^{\frac12} \cdot r^{-\frac43} + \sqrt{\log n})$ from the
  corners, $m$ is also within
  $O(n^{\frac12} \cdot r^{-\frac43} + \sqrt{\log n})$ from the
  geometric center of $\mathcal{S}$.  Combining these bounds, we
  conclude that $v_c$ is located within
  $O(n^{\frac12}\cdot r^{-\frac23} + \sqrt{\log n})$ from the
  geometric center.
\end{proof}

\subparagraph*{Analysis of iFUB.}%
We briefly explain how the algorithm proceeds after selecting a
central vertex $c$, see also \cite{ifub}.  Let $v_1, \dots, v_n = c$
be an ordering of the vertices sorted in descending order of their
distance to $c$.  Such an ordering is easily obtained after running a
BFS from $c$.  To find the diameter, the algorithm computes
$\ecc(v_i)$ for each vertex $v_i$ in this sequence and maintains the
largest found eccentricity as a lower bound, i.e.,
$L_i = \max_{0≤j≤i} \ecc(v_j)$.  The algorithm stops and reports $L_i$
as the diameter, once $2 \cdot \dist(c,v_i) \le L_i$.  To see why this
is correct, note that there exists a diametrical vertex $s$ with
$\dist(c,s) ≥ \ceil{\frac{\diam_G}{2}}$ and that the stopping criterion
ensures that such a vertex has been processed.

As discussed in the introduction of this paper, the running time of
iFUB depends on the choice of the central vertex and the metric
structure of the graph.  To be exact, the running time depends on the
number of vertices with distance at least half the diameter from $c$.
This has already been observed and used in the
literature~\cite{axiomatic_borassi_2017}, but to the best of our
knowledge not formally proved.  We consequently give a complete
argument below.

\begin{lemma}\label{lem:ifub_explored}
  Let $G$ be a graph with diameter $D$ and let $c$ be the central
  vertex for iFUB.  Then, iFUB explores every vertex with distance at
  least $\ceil{\frac{ D }{ 2 }}+1$ from $c$ and every explored vertex
  has distance at least $\ceil{\frac{ D }{ 2 }}$ from $c$.
\end{lemma}
\begin{proof}
  We begin with the first direction, i.e., we show that a vertex $v$
  with $\dist(c,v) ≥ \ceil{\frac{D}{2}}+1$ is explored by iFUB.  Let
  $L_i \le D$ be the value of the lower bound at the time when iFUB
  decides whether to explore $v$. Then we have
  $2 \dist(w,c) \ge \frac{D}{2} + 2 > L_{i}$, i.e., iFUB explores $w$.

  For the other direction, let $w$ be a vertex that is explored by
  iFUB.  Then, at the time when iFUB explores $w$, we have
  $2 \dist(c,w) > L_i$.  If $L_i = D$, this concludes the proof, as we
  have $\dist(v,w) ≥ \ceil{D/2}$.  If otherwise $L_i < D$, then no
  diametrical vertex has been explored yet.  However, there is a
  diametrical vertex $x$ with $\dist(c,x) ≥ \ceil*{D/2}$.  As $x$ has
  not yet been explored when $w$ is explored, we have
  $\dist(c,w) \ge \dist(c,x)$ and thus also
  $\dist(c,w) \ge \ceil*{D/2}$.
\end{proof}


We already analyzed the 2-sweep and showed that the central vertex is
likely to be located close to the geometric center of the square.  It
remains to show that in this case iFUB does not perform too many BFS.

\begin{lemma}\label{lem:ifub:center}
  Let $G\sim\mathcal{G}(\mathcal{S}, n, r)$ be a square random geometric
  graph with $r∈ω(\log^{3/4} n)$ and let $v_c$ be a vertex with
  geometric distance $h∈o(\sqrt{n})$ from the geometric center of
  $\mathcal{S}$.  Then, with $v_c$ as central vertex iFUB performs at
  most $O(h²+nr^{-\frac83})$ BFS runs, asymptotically almost surely.
\end{lemma}
\begin{proof}
  With $v_c$ chosen as the central vertex, iFUB performs a BFS for
  every vertex $w$ with $\dist_G(v_c, w) \ge \frac{\diam_G}{2}$.  We
  show that vertices close to the geometric center of $\mathcal{S}$ do
  not have graph distance at least $\frac{\diam_G}{2}$.  Conversely, the
  vertices from which a iFUB runs a BFS lie in regions far from the
  geometric center.  We show that these regions do not contain many
  vertices.  We condition on the stretch event, which holds
  a.a.s.\ (\cref{lem:stretch}).

  We choose $d = \frac{\sqrt{n}}{√2} - x$ for some $x>0$ to be
  specified later and consider a vertex $u$ with geometric distance at
  most $d$ from the center of $\mathcal{S}$.  Then
  $\dist_\mathbb{E}(v_c,u) \le \frac{\sqrt{n}}{√2} - x + h$ and thus
  $\dist_G(v_c,u) \le \parens*{\frac{\sqrt{n}}{√2} - x + h}r^{-1} (1+s)$ for
  $s∈O(r^{-\frac43})$.  By \cref{lem:rgg_diam} we have
  $\diam_G \ge \frac{\sqrt{2n}}{r}$.  Thus, we can choose
  $x∈Θ\parens*{h+\sqrt{n}s}$ such that
  $\dist_G(v_c,u) < \frac{\diam_G}{2}$.

  As iFUB only runs BFS from vertices with distance at least
  $\frac{ \diam_G }{2}$ from $v_c$ (\cref{lem:ifub_explored}), this
  means that any such vertex has distance at least
  $\frac{\sqrt{n}}{√2} - x$ from the geometric center of
  $\mathcal{S}$.  Comparing this with \cref{lem:square_geom_fringe},
  we get that any such vertex has distance at most $O(h + \sqrt{n}s)$
  from a corner of $\mathcal{S}$ and thus lies in a region with area
  $O(h^2 + n\cdot s^2)$.  The number of vertices in this region is in
  $O(h²+n\cdot s²)$ with high probability by
  \cref{lem:area_vertex_count}, which concludes the proof.
\end{proof}

Together with \cref{lem:ifub:center} this results in the following
running time bound, which is (truly) subquadratic for (polynomially)
growing $r$.

\begin{lemma}\label{lem:ifub:rgg}
  Let $G\sim\mathcal{G}(\mathcal{S}, n, r)$ be a square random geometric
  graph with $r∈ω(\log^{3/4} n)$.  Then, asymptotically almost surely
  2-sweep iFUB has running time in $O((nr^{-\frac43}+ \log n)m)$.
\end{lemma}
\begin{proof}
  By \cref{lem:ifub:2sweep}, the central vertex $v_c$ chosen after the
  2-sweep has distance at most
  $h ∈ O(n^{\frac12}r^{-\frac23} + \log^{\frac12} n)$ from the geometric center
  of $\mathcal{S}$, asymptotically almost surely.  Thus by
  \cref{lem:ifub:center} iFUB performs only
  \[
    O(h²+n\cdot r^{-\frac83}) = O(nr^{-\frac43}+\log n)
  \]
  BFS runs, asymptotically almost surely.
\end{proof}


We also want to show a lower bound for the running time on torus RGGs.
For this, we use that the iFUB algorithm performs a BFS for all
vertices with distance at least $\ceil{\frac{ \diam_G }{2}}+1$ from
the central vertex $v_c$.  By observing that on torus RGGs there are
many vertices at such a distance from any chosen central
vertex, this gives us a linear lower bound for the number of BFS runs.

\begin{lemma}
  Let $G\sim\mathcal{G}(\mathcal{T}, n, r)$ be a torus random geometric
  graph with $r∈ω(\log^{3/4} n)$.  Then, asymptotically almost
  surely, for every central vertex, iFUB performs $Ω(n)$ BFS runs.
\end{lemma}
\begin{proof}
  Let $v_c$ be the central vertex for iFUB.  Then iFUB performs a BFS
  for any vertex $w$ with
  $\dist_G(v_c,w) \ge \ceil{\frac{\diam_G}{2}}+1$
  (\cref{lem:ifub_explored}).  By \cref{lem:torus_rgg_diam}, we have
  $\diam_G \le \frac{\sqrt{n}}{\sqrt{2}r} (1+s)$ with
  $s∈O(r^{-\frac43})$, asymptotically almost surely.  By
  \cref{eq:stretch-trivial}, for a vertex $w$ with
  \[
    \dist_\mathbb{E}(v_c, w) \ge \frac{\sqrt{n}}{2\sqrt{2}} (1+1.1s) + 2r
  \]
  we have
  $\dist_G(v_c, w) \ge \frac{\sqrt{n}}{2\sqrt{2}r} (1+1.1s) + 2> \ceil{\frac{ \diam_G }{ 2 }}+1$.
  It remains to show that many vertices have such a distance from
  $v_c$.  The geometric disk of radius
  $\frac{\sqrt{n}}{2\sqrt{2}} (1+1.1s)+2r$ around $v_c$ has area
  \[
\frac{π n}{8}\left(1+2.2s+1.21s^{2}+\frac{16r}{\sqrt{2n}}(1+1.1s)+\frac{32r^{2}}{n}\right),
  \]
where $s,s²,\frac{r}{n},$ and $r²/n$ are all in $o(1)$.  The entire
torus $\mathcal{T}$ has area $n$ and thus the region of $\mathcal{T}$
where vertices are chosen as BFS sources by iFUB has area at least
$(1-\frac{π}{8})n (1-o(1)) \ge 0.6 n (1-o(1))$.  As the number of
vertices within such a region is sufficiently concentrated by
\cref{lem:area_vertex_count}, this means that iFUB performs $Ω(n)$ BFS
runs and thus has a running time in $Ω(nm)$.
\end{proof}

Together, the above two lemmas give the following.

\ifubfast*

\section{Conclusion}
\label{sec:conclusion}

In this paper we give a set of natural deterministic properties
allowing for efficient diameter computation and demonstrate that these
properties a.a.s.\ hold on square and torus RGGs.  We note that our
formulation of the properties is not the only possible one, but
represents a trade-off between simplicity and generality.  To show
this, we point out multiple possible generalizations for the
assumptions used in our algorithm.

In \cref{item:prop-diam-partners} we demand that for all $x>0$ and
every vertex $v$, the $x$-diametric partners lie in $O(1)$ balls of
radius $O(x + d_{\mathrm{local}})$.  Here, the linear dependence on
$x$ was mostly chosen for its simplicity.  By considering how the
property is used (e.g. \cref{lem:candidate_bound}) one can see that
this dependence can be significantly relaxed.  For instance, one could
demand that the radius of the balls has some arbitrary non-decreasing
dependency $f(x)$ on $x$, or even depends on $n$ and $x$ as $f_n(x)$.
Then in \cref{thm:properties_algo}, the new requirement is that blocks
of size $k$ need to have diameter in $Ω(f_n(x))$.  Similarly, instead
of requiring a constant number of balls, one could also specify the
number of balls as a parameter, which then appears as an additional
factor in the running time of \cref{thm:properties_algo}.  Both of
these generalizations apply analogously to
\cref{item:prop-few-corners}.  Moreover, for
\cref{item:prop-fragmentation} one could allow a non-constant
parameter for the number of intersecting blocks, which then appears as
an additional factor in the number of candidate pairs and thus the
running time.

We also want to point out some directions for improvement regarding
the analysis on random geometric graphs.  For the application of our
algorithm on RGGs we assumed that the algorithm receives the graph
along with a suitable recursive partition, see \cref{thm:algo_rgg}.
While \cref{sec:rgg-nice} demonstrates that such a partition is
obtained very easily by subdividing the graph along its geometry, it
would be interesting to also give an algorithm that finds a suitable
partition using only a combinatorial representation of the graph
without coordinates.  We believe that a simple approach based on graph
Voronoi diagrams should already work, but it seems like showing tight
bounds for the size of (recursive) graph Voronoi separators in random
geometric graphs is very challenging.  Can this challenge be overcome
or is it maybe possible to find a different approach that is easier to
analyze?  Moreover, much of our analysis hinges on the stretch bounds
(\cref{lem:stretch,lem:torus_stretch}).  It is not clear how tight
these are and whether polynomially growing average degree is really
necessary.  It would not be too surprising if RGGs with constant
average degree also have local diametric partners
(\cref{item:prop-diam-partners}).  Finally, our running time
analysis for iFUB on square RGGs is likely pessimistic and better
stretch bounds or more generally a better understanding of the
distribution of graph distances can be expected to improve this.


%
%
%



\bibliography{paper}

\appendix

\section{Asymptotics of Properties from Section~\ref{sec:intro}}
\label{sec:asymptotics_of_properties}

In order to talk about asymptotic running times of algorithms, one needs
to consider infinite families of inputs.  As the definitions of \cref{item:prop-diam-partners,item:prop-few-corners,item:prop-small-sep,item:prop-size-dependent-diam,item:prop-fragmentation}
in \cref{sec:intro} do not make their asymptotic interpretations
explicit, we provide formal definitions of the properties defined in
\cref{sec:intro} in this section.  Let
$\mathcal{G} = \{G_1, G_2, \dots\}$ be an infinite family of graphs.

For the first property, local diametric partners, the asymptotic
interpretation is straightforward.  The only important detail is that
the constants hidden by the big $O$-notation may not depend on
individual graphs.

\begin{property}[local diametric partners]
  We say that $\mathcal{G}$ has \emph{$d_{\textrm{local}}$-local
    diametric partners}, if there exist positive integer constants $a, b, c$ such
  that for every graph $G∈\mathcal{G}$, every vertex $v∈V(G)$, and
  every positive integer $x$, there exists at most $a$ vertices
  $w_1, \dots, w_a∈V(G)$ such that the union of their closed
  $b \cdot \parens*{x + d_{\mathrm{local}}} + c$ neighborhoods
  contains every $x$-diametric partner of $v$ in $G$.
\end{property}

The asymptotic interpretation of the second property, few corners, is analogous.

\begin{property}[few corners]
  We say that $\mathcal{G}$ has \emph{$d_{\mathrm{corner}}$-few
    corners}, if there exist positive integer constants $a, b, c$ such
  that for every graph $G∈\mathcal{G}$ and every positive integer $x$
  there exist up to $a$ vertices $w_1, \dots, w_a ∈V(G)$ such that the
  union of their closed
  $b \cdot \parens*{x + d_{\mathrm{corner}}} + c$ neighborhoods
  contains every vertex with at least one $x$-diametric partner in
  $G$.
\end{property}

The remaining properties also depend on recursive partitions.  For
each graph $G_i ∈ \mathcal{G}$, let $\mathcal{P}_i$ be a recursive
partition.  Then
$\mathcal{I} = \{(G_1, \mathcal{P}_1), (G_2, \mathcal{P}_2), \dots \}$
forms an infinite family of graphs together with recursive partitions.
With this the asymptotic interpretation of the third property is again
straightforward, again with the only important detail being that the
constants hidden in the big $O$-notation must be universal for the
family of instances.

\begin{property}[small separators]
  We say that $\mathcal{I}$ has \emph{$(α,β)$-small separators}, if
  there exist positive integer constants $b, c$ such that for every
  $(G, \mathcal{P}) ∈ \mathcal{I}$ and every block $B$ induced by
  $\mathcal{P}$ on $G$, the separator of $B$ has size at most
  $b \cdot \parens*{n^α \cdot |B|^β} + c$.
\end{property}

For the fourth property it is important to only compare the sizes and
diameters of blocks of the same graph, as across graphs similarly
sized blocks are allowed to have different diameter.

\begin{property}[size-dependent diameters]
  We say that $\mathcal{I}$ has \emph{size-dependent diameters}, if
  there exist constants $b,c,d,e > 1$ such that for every
  $(G, \mathcal{P})∈\mathcal{I}$ and every two blocks $A, B$ induced
  by $\mathcal{P}$ on $G$, we have $|A| \le b \cdot |B| + c$ if and
  only if $\diam_{G[A]} \le d \cdot \diam_{G[B]} + e$.
\end{property}

For the fifth property, the interpretation is again straightforward.

\begin{property}[low fragmentation]
  We say that $\mathcal{I}$ has \emph{low fragmentation} if there
  exist positive integer constants $a,b$ such that for every
  $(G, \mathcal{P})∈\mathcal{I}$, every vertex $v∈V(G)$ and every
  integer $x$, the closed $x$ neighborhood of $v$ intersects at most
  $a$ blocks $B$ of $\mathcal{P}$ with diameter $\diam_{G[B]}$ between
  $\frac{1}{b} x$ and $bx$.
\end{property}

\end{document}